\documentclass[lettersize,journal]{IEEEtran}
\usepackage{amsfonts}
\usepackage{amsmath}  
\usepackage{amssymb}  
\usepackage{amsthm} 
\usepackage{algorithmic}
\usepackage{algorithm}
\usepackage{array}
\usepackage[caption=false,font=normalsize,labelfont=sf,textfont=sf]{subfig}
\usepackage{textcomp}
\usepackage{stfloats}
\usepackage{url}
\usepackage{verbatim}
\usepackage{graphicx}
\usepackage{xcolor}
\usepackage{cite}
\usepackage{booktabs}
\usepackage{multirow}
\usepackage{hyperref}
\usepackage{verbatim}
\usepackage{graphicx}   
\usepackage{caption}    
\usepackage{subcaption} 
\usepackage{enumitem}

\usepackage[dvipsnames,svgnames,x11names]{xcolor}
\usepackage{tikz}
\usepackage{moresize}
\usepackage{pgfplots}
\usepackage{wrapfig}
\pgfplotsset{compat=1.17}
\pgfplotstableset{col sep=comma}
\pgfplotsset{compat=1.17}
\pgfplotstableset{col sep=comma}
\tikzset{every mark/.append style={scale=1.5, solid}, font=\footnotesize}
\pgfplotsset{
    width=1.05\textwidth,
    tick label style={font=\small},
    label style={font=\small},
    legend style={
        font=\footnotesize ,  
        inner xsep=1pt,
        inner ysep=1pt,
        nodes={inner sep=1pt}},
    legend cell align=left,
	every axis/.append style={line width=0.5pt},
	every axis plot/.append style={line width=1.25pt},
    every axis y label/.append style={yshift=-5pt}
}

\theoremstyle{definition} 

\newtheorem{mydefinition}{\bf Definition}
\newtheorem{myproposition}{\bf Proposition}

\usepackage{stmaryrd}
\theoremstyle{remarkstyle}
\newtheorem{remark}{Remark}

\input{mysymbol.sty}

\def \MSE {\mathrm{MSE}}

\newcommand{\alna}[1]{\begin{alignat}{3}&#1&\end{alignat}}
\def \bi {\begin{itemize} \item}
\def \ei {\end{itemize}}
\def \im {\item}

\definecolor{clrsamp}{HTML}{000000}
\definecolor{clrcorr}{HTML}{5f5695}
\definecolor{clrperi}{HTML}{47407d}
\definecolor{clrmaspat}{HTML}{8faeff}
\definecolor{clrmaspec}{HTML}{4671df}
\definecolor{clrarspat}{HTML}{e865a5}
\definecolor{clrarspec}{HTML}{ba1968}
\definecolor{clrwiener}{HTML}{e19b00}
\definecolor{clrnoisy}{HTML}{785ef0}

\definecolor{clra0}{HTML}{47407d}
\definecolor{clra0001}{HTML}{760039}
\definecolor{clra001}{HTML}{ba1968}
\definecolor{clra01}{HTML}{dc267f}
\definecolor{clra1}{HTML}{f3a3ca}

\newtheorem{theorem}{Theorem}
\newtheorem{corollary}[theorem]{Corollary}

\begin{document}

\title{Stationarity and Spectral Characterization of Random Signals on Simplicial Complexes}

\author{
    Madeline Navarro,~\IEEEmembership{Student Member,~IEEE,} 
    Andrei Buciulea,~\IEEEmembership{Member,~IEEE,} 
    Santiago Segarra,~\IEEEmembership{Senior Member,~IEEE,} 
    and Antonio G. Marques,~\IEEEmembership{Senior Member,~IEEE}

    \thanks{
        Madeline Navarro and Santiago Segarra are with the ECE Department, Rice University. Andrei Buciulea and Antonio G. Marques are with the Department of Signal Theory and Communications, King Juan Carlos University. 
        
        This paper is supported by the Spanish AEI (AEI/10.13039/501100011033) grant PID2023-149457OB-I00, and the Community of Madrid via IDEA-CM (TEC-2024/COM-89) and the Ellis Madrid Unit.
    }
}


\maketitle

\begin{abstract}
It is increasingly common for data to possess intricate structure, necessitating new models and analytical tools.
Graphs, a prominent type of structure, can encode the relationships between any two entities (nodes).
However, graphs neither allow connections that are not dyadic nor permit relationships between sets of nodes.
We thus turn to \emph{simplicial complexes} for connecting more than two nodes as well as modeling relationships between simplices, such as edges and triangles.
Our data then consist of signals lying on \emph{topological spaces}, represented by simplicial complexes.
Much recent work explores these topological signals, albeit primarily through deterministic formulations.
We propose a \emph{probabilistic} framework for random signals defined on simplicial complexes.
Specifically, we generalize the classical notion of \emph{stationarity}.
By spectral dualities of Hodge and Dirac theory, we define stationary topological signals as the outputs of topological filters given white noise.
This definition naturally extends desirable properties of stationarity that hold for both time-series and graph signals.
Crucially, we properly define topological \emph{power spectral density} (PSD) through a clear spectral characterization.
We then discuss the advantages of topological stationarity due to spectral properties via the PSD.
In addition, we empirically demonstrate the practicality of these benefits through multiple synthetic and real-world simulations.
\end{abstract}


\section{Introduction}

As data science grows progressively more complex, modern tasks increasingly involve data residing on irregular domains.
While graphs are popular models for complex systems, their limitation to pairwise interactions is often insufficient, motivating higher-order connections that relate multiple entities.
When polyadic relationships are also associated with data, we can exploit simplicial complexes as concise yet expressive domains for signals supported on nodes, edges, triangles, and so on~\cite{bick2023higher, salnikov2018simplicial}.
Recent years have seen the development of a rich set of deterministic tools for topological (simplicial) signals~\cite{barbarossa2020topological, schaub2021signal}.
However, from biological modeling to traffic, finance, and network inference, practical applications typically do not come with perfectly measured data~\cite{isufi2024topological}.
Existing deterministic tools may fail to account for uncertainties arising from noise, missing values, or model discrepancies.
It behooves us to develop tractable \emph{probabilistic} models and inference tools to reliably draw conclusions from high-dimensional data.

As a natural starting point, we focus on \emph{stationarity}, the foundation of spectral analysis and estimation in time series and random fields~\cite{hayes2009statistical,stoica2005spectral}.
Stationarity has also been well established for signals in the graph domain~\cite{girault2015stationary,girault2015translation,EPFL16stationary,marques2017stationary}.
For example, past works define weak stationarity for random graph signals analogously to time series~\cite{sandrymouraspg_tsp14freq,segarra2017optimal}.
As in the time domain, these definitions of graph stationarity result in convenient simplifications of spectral behavior, leading to the generalization of essential definitions in signal processing (SP), such as the power spectral density (PSD) or spectral convolutions.

Here, we consider signals that reside on simplicial complexes, which subsume graph signals yet permit multiway interactions~\cite{bick2023higher,salnikov2018simplicial}.
Through the seminal Hodge decomposition~\cite{barbarossa2020topological,lim2020hodge,schaub2021signal,yang2022simplicial,isufi2022convolutional}, a simplicial signal of any given order can be decomposed into gradient, curl, and harmonic components.
Not only is this algebraically convenient, but it also reflects structure observed in real data, such as the nearly divergence-free nature of flows in transportation networks or arbitrage-free currency exchanges possessing curl-free edge signals~\cite{yang2022simplicial,schaub2018flow,jia2019graph,liu2023unrolling}.
Moreover, when data exist in multiple dimensions, such as signals on both nodes and edges, we need not process them separately; the Dirac operator couples signals across orders, allowing \emph{multiorder} spectral processing~\cite{bianconi2021topological,baccini2022weighted}.
Both Hodge and Dirac operations have led to the topological extension of many classical concepts, including convolutional architectures, trend filtering, neural networks, autoregressive (AR) modeling, recovery methods, random walks, and Gaussian models~\cite{isufi2024topological,barbarossa2020topological,schaub2021signal,yang2022simplicial,yang2022simplicialtrend,yang2022simplicialconvolutional,battiloro2024generalized,krishnan2023simplicial,reddy2024recovery,schaub2020random,marinucci2025simplicial}.

Despite this progress, probabilistic modeling for \emph{random} simplicial signals is still nascent.
Some well-known probabilistic models have seen adaptations to simplicial structures, such as diffusion models and Gaussian processes~\cite{yang2025topological,yang2023hodge}
However, we require more general stochastic frameworks that are computationally tractable for uncertainty modeling, spectral analysis, and statistically sound inference across simplicial orders.
In particular, we seek models that (i) respect Hodge/Dirac spectral geometry, (ii) provide clear spectral characterizations, and (iii) lead to principled estimators and filters suitable for practical setups, such as data-scarce settings.

\noindent \textbf{Our approach.} 
Motivated by the graph setting and properties of Hodge or Dirac-based topological spectra, we develop a probabilistic framework for \emph{stationary topological random signals}.
In Section~\ref{S:StatDefinition}, we define a stationary topological process as the output of a topological filter given white noise.
We also show convenient characterizations of topological signal behavior through the spectral properties of the simplicial complex, either for a single, given order or in multiorder scenarios.
Our definitions are mathematically straightforward, aligning with stationarity for graphs or time-series, yet we provide richer modeling capacity through the multiway relationships permitted by simplicial complexes.

More specifically, we make the following contributions:
\begin{itemize}[left= 2pt .. 10pt, noitemsep]
\item[1)] We introduce a definition of \emph{topological stationarity} for random simplicial signals and derive the associated PSD.
\item[2)] We develop and analyze both nonparametric and parametric methods to estimate either the PSD or covariance of stationary random topological signals.
\item[3)] We generalize classical statistical definitions and tools (including the Wiener filter)  to the topological domain and empirically illustrate their use in canonical SP tasks such as signal denoising.
\end{itemize}

\noindent\textbf{Organization.}
We first review preliminaries on simplicial complexes, Hodge and Dirac operators, and spectral decompositions in Sec.~\ref{S:prelim}.
Then, in Sec.~\ref{S:StatDefinition}, we formalize stationarity for topological signals and relate it to prior formulations.
In Sec.~\ref{S:estimators}, we present multiple estimators for the covariance matrix and PSD of stationary topological signals, while in Sec.~\ref{S:problem}, we demonstrate how to generalize classical tools such as the Wiener filter to the topological domain.
Sec.~\ref{S:exps} reports numerical experiments validating the use of these tools for estimating the covariance and PSD of topological signals under stationarity.
Finally, we conclude the paper in Sec.~\ref{S:conclusion}.

\section{Topological SP Preliminaries} \label{S:prelim}

We begin with a brief overview of the fundamentals of topological SP to lay the groundwork for stationary simplicial signals, defined in the sequel.
A more detailed version of this review can be found in \cite{liu2025matched}.

\begin{figure}[t]
    \centering
    \vspace{-.2cm}
        \hspace{-0.5cm}
        \begin{minipage}[t]{0.5\linewidth}
            \centering
            \includegraphics[width=40mm]{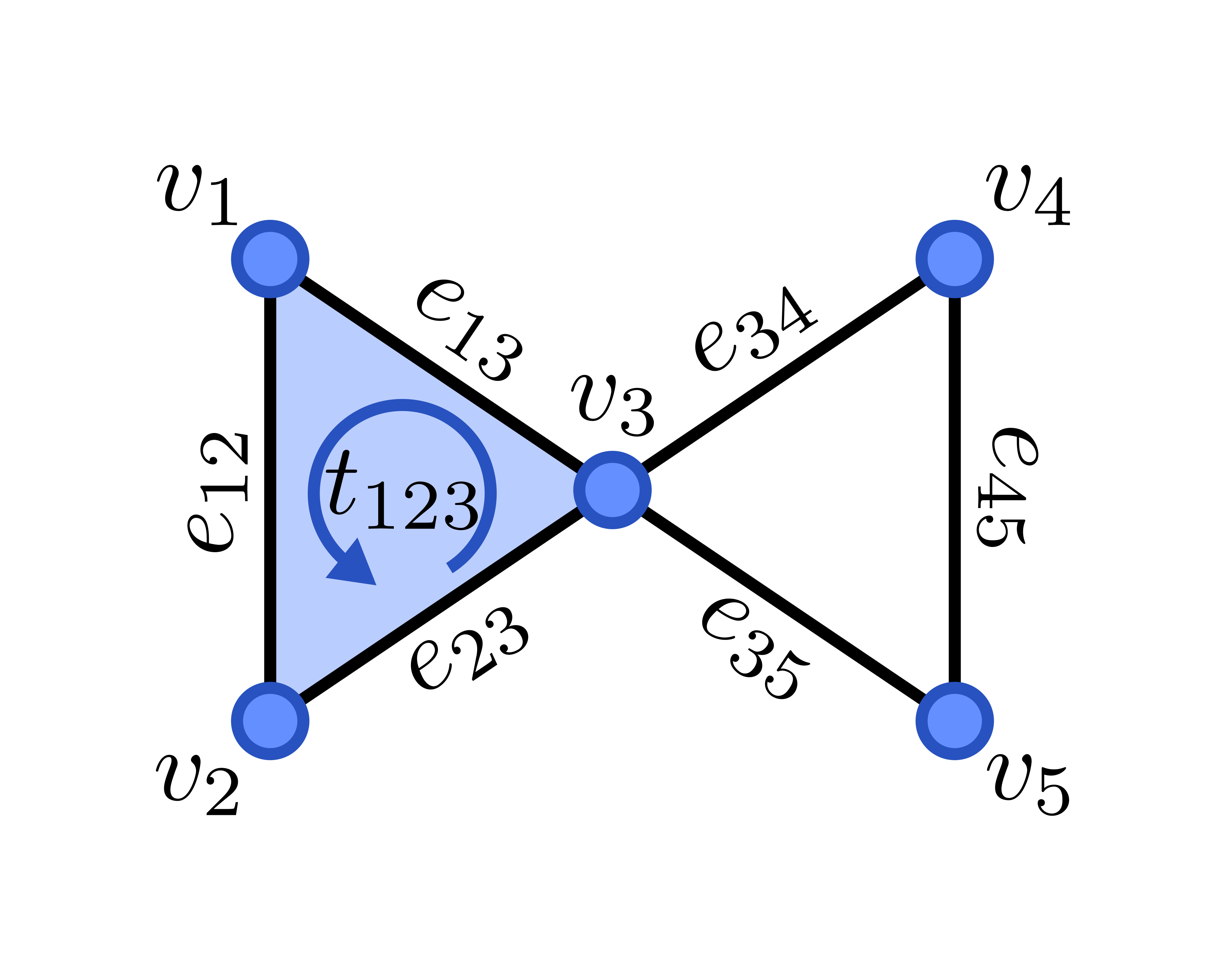}
            \centerline{\vspace{-.2cm}\small(a) $2$-order example.}
        \end{minipage}%
        \hspace{-.0cm}
        \begin{minipage}[t]{0.5\linewidth}
            \centering
            \includegraphics[width=40mm]{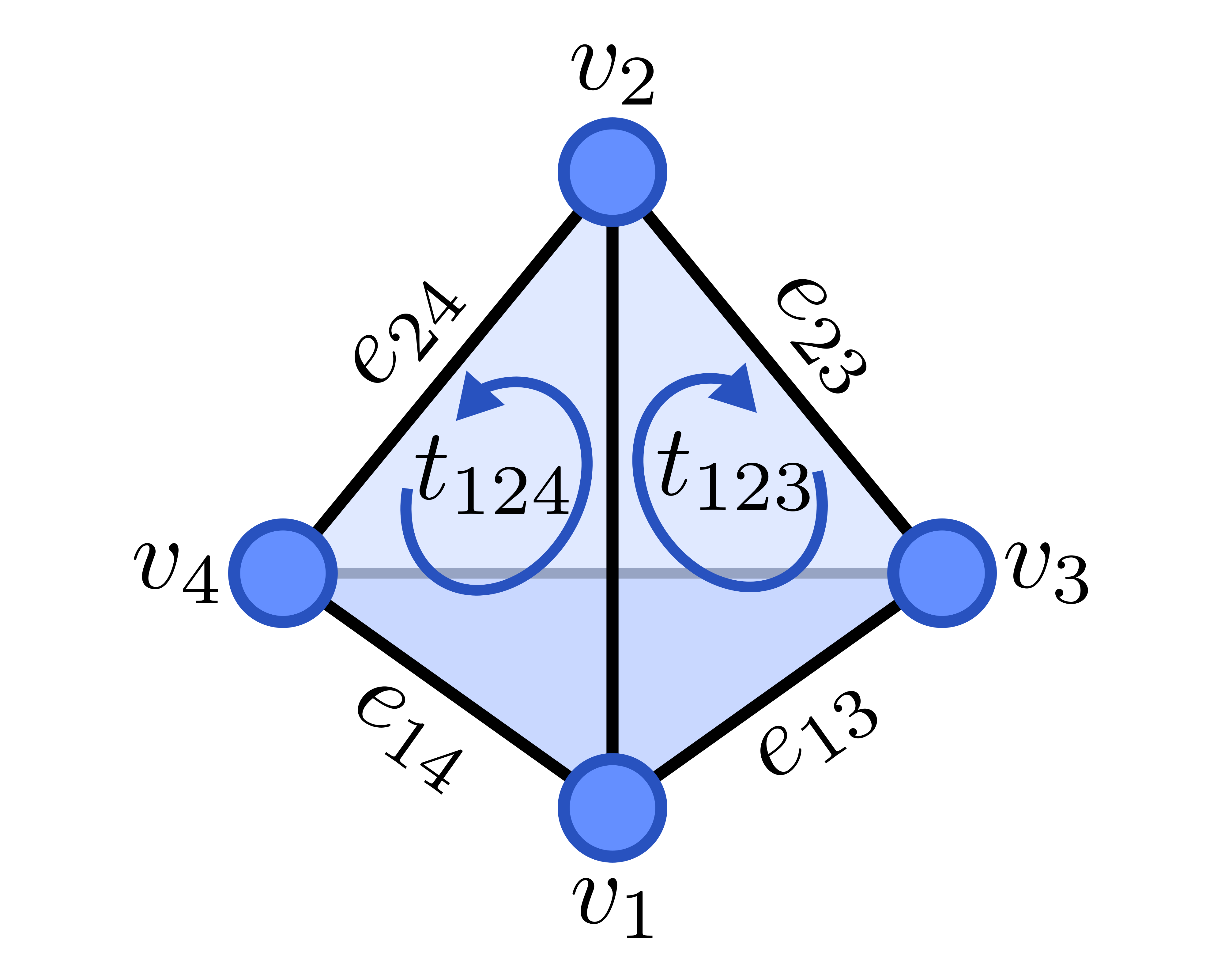}
            \centerline{\vspace{-.2cm}\small(b) $3$-order example.}
        \end{minipage}
        \hspace{-.2cm}
    \caption{Examples of simplicial complexes of different orders.
    (a) A simplicial complex of order $2$ with one triangle $\{1,2,3\}$. 
    The edge $\{1,2\}$ is \emph{aligned} with $\{1,2,3\}$ but $\{1,3\}$ is \emph{anti-aligned}.
    (b) A simplicial complex of order $2$ with one tetrahedron $\{1,2,3,4\}$.
    The triangle $\{1,2,4\}$ is \emph{aligned} with $\{1,2,3,4\}$ but $\{1,2,3\}$ is \emph{anti-aligned}.
    }
    \vspace{-.3cm}
    \label{fig:sc_example}
\end{figure}

\subsection{Simplicial Complexes and Simplicial Signals}

A simplicial complex consists of multiway relationships among a set $\ccalV$ of $N_0$ vertices.
In particular, the simplicial complex $\ccalP^K$ of order $K \leq N_0$ contains simplices, that is, subsets of $\ccalV$, where the cardinality of the largest simplex in $\ccalP^K$ is $K+1$.
Moreover, if $\ccalW \in \ccalP^K$ contains $k+1$ vertices, then $\ccalW$ is a $k$-simplex.
We may write the set of $k$-simplices in $\ccalP^K$ as $\{ \ccalW^k_n \}_{n=1}^{N_k}$, $N_k$ is the number of $k$-simplices, where the total number of simplices is $N = \sum_{k=1}^K N_k$.
Thus, we may define $\ccalP^K$ as the collection $\{ \ccalW^0_n \}_{n=1}^{N_0}$, $\{ \ccalW^1_n \}_{n=1}^{N_1}$, \dots, $\{ \ccalW^K_n \}_{n=1}^{N_K}$, satisfying the \emph{inclusion property}: For every simplex $\ccalW \in \ccalP^K$, all of its subsets are also in $\ccalP^K$, that is, $\ccalW' \in \ccalP^K$ for all $\ccalW' \subset \ccalW$.
For a given $k$-simplex $\ccalW^k$, the $(k-1)$-simplex $\ccalW^{k-1}$ is a \emph{face} of $\ccalW^k$ if $\ccalW^{k-1} \subset \ccalW^k$, while the $(k+1)$-simplex $\ccalW^{k+1}$ is a \emph{coface} of $\ccalW^k$ if $\ccalW^{k+1} \supset \ccalW^k$.
While a seemingly abstract definition, the simplicial complex allows intuitive representations of connected data in Euclidean space.
Fig.~\ref{fig:sc_example}a exemplifies a simplicial complex of order $2$, where $0$-simplices denote vertices, $1$-simplices edges, and $2$-simplices triangles, while Fig.~\ref{fig:sc_example}b shows a simplicial complex of order $3$ with a $3$-simplex as a tetrahedron.

We next introduce operators that relate signals across adjacent dimensions.
To this end, we require reference orientations of any simplex in $\ccalP^K$.
We let the canonical representation of any simplex $\ccalW \subseteq \ccalV$ in $\ccalP^k$ be the subset of $\ccalV$ ordered by vertex indices.
For example, $\{ 1,2,4 \}$, $\{2,1,4\}$, and $\{ 2,4,1 \}$ all correspond to the same $2$-simplex, canonically represented as $\{1,2,4\}$.
Observe that $\{2,1,4\}$ is anti-aligned with $\{1,2,4\}$ since the directional edge flow $2 \,{\small \rightarrow}\, 1 \,{\small \rightarrow}\, 4 \,{\small \rightarrow}\, 2$ opposes the canonical $1 \,{\small \rightarrow}\, 2 \,{\small \rightarrow}\, 4 \,{\small \rightarrow}\, 1$.
We can then introduce incidence matrices $\bbB_k \in \{ -1, 0, 1 \}^{ N_{k-1} \times N_k }$ for all $k = 1,\dots,K$, which link simplices of adjacent orders $k-1$ and $k$~\cite{barbarossa2020topological}.
If $[\bbB_k]_{nm} \neq 0$, then $\ccalW_m^k \in \ccalP^K$ exists and $\ccalW^{k-1}_n \subset \ccalW_m^k$, and the sign of $[\bbB_k]_{nm}$ denotes if $\ccalW^{k-1}_n$ and $\ccalW^k_m$ are aligned or anti-aligned with respect to their reference orientations, which is demonstrated in Fig.~\ref{fig:sc_example}.
Using these matrices, we introduce \emph{Hodge Laplacians} for each dimension~\cite{lim2020hodge}
\begin{equation}\label{eq:hodge_lap}
\begin{aligned}
	& \mathbf{L}_0=\mathbf{B}_1 \mathbf{B}_1^\top,~\mathbf{L}_K=\mathbf{B}_K^\top \mathbf{B}_K,~\text{and} \\
	& \mathbf{L}_k=\mathbf{L}_{k,\ell}+\mathbf{L}_{k,u}~\text{for}~k=1, \ldots, K-1,\\
	&\hspace{1cm}\text{with}~\mathbf{L}_{k,\ell}=\mathbf{B}_k^\top \mathbf{B}_k,~\text{and}~\mathbf{L}_{k,u}=\mathbf{B}_{k+1} \mathbf{B}_{k+1}^\top.
	\quad  
\end{aligned}
\end{equation}
Each dimension $1 \leq k \leq K-1$ is associated with a lower Laplacian $\bbL_{k,\ell} := \bbB_k^\top \bbB_k$ and an upper Laplacian $\bbL_{k,u} := \bbB_{k+1} \bbB_{k+1}^\top$, which connect $k$-simplices via shared faces and cofaces, respectively.
The Hodge Laplacians in~\eqref{eq:hodge_lap} describe how simplices of a given order are connected, but we will also apply them as topological signal operators.

To this end, we must first introduce \emph{simplicial signals}, that is, real-valued data on simplicial complexes.
First, we consider data at a single order.
A $k$-simplicial signal (abbreviated as $k$-signal) is represented by a vector $\bbs_k = [s_1^k, \dots, s_{N_k}^k]^\top \in \reals^{N_k}$, with the entry $s_n^k$ denoting the value of the $n$-th $k$-simplex~\cite{barbarossa2020topological}.
The sign of $s_{n}^k$ indicates if the signal orientation is aligned or anti-aligned with the reference orientation of the $n$-th simplex, which we exemplify in Fig.~\ref{fig:sc_example}.
We then define a complete \emph{simplicial complex signal} as the concatenation of signals across all orders
\alna{
    \bbs = \big[
        \bbs^{0\top}, \dots, \bbs^{K\top}
    \big]^\top \in \reals^N.
\label{eq:simplicial_complexes_signal}}

\subsection{Hodge Decomposition}\label{Ss:hodge_decomp}

A critical fact of the incidence matrices is that $\bbB_k \bbB_{k+1} = \bbzero$~\cite{hatcher2002algebraic}.
Then, by the definition of the Hodge Laplacians in~\eqref{eq:hodge_lap}, we arrive at the \emph{Hodge decomposition}, which allows us to decompose simplicial signal spaces into three orthogonal yet meaningful subspaces.
By the orthogonality of $\operatorname{span}(\bbB_k^\top)$ and $\operatorname{span}(\bbB_{k+1})$, we can uniquely decompose any $k$-signal $\bbs^k$ as
\alna{
    \bbs^k = \bbs^k_{\rm G} + \bbs^k_{\rm C} + \bbs^k_{\rm H},
\label{eq:hodge_components}}
where $\bbs^k_{\rm G} \in \operatorname{span}(\bbB_k^\top)$, 
$\bbs^k_{\rm C} \in \operatorname{span}(\bbB_{k+1})$, and
$\bbs^k_{\rm H} \in \operatorname{kernel}(\bbL_{k})$
lie in mutually orthogonal subspaces corresponding to gradient, curl, and harmonic components, respectively~\cite{lim2020hodge}.
Intuitively, given some $k$-signal $\bbs$, we interpret each component as follows~\cite{yang2022simplicial}.
\begin{itemize}[left= 2pt .. 10pt, noitemsep]
\im {\bf Gradient:}
    The transformation $\bbB_{k+1}^\top \bbs$ maps values in $\bbs$ to the $(k+1)$-cofaces of all $k$-simplices, representing the circulation of $k$-signal values around $(k+1)$-simplices, inducing a discrete curl operation.
\im {\bf Curl:}
    We map $\bbs$ to the $(k-1)$-faces of each $k$-simplex via $\bbB_k\bbs$, 
    This sums $k$-signal values at each $(k-1)$-simplex, denoting the net flow and therefore inducing a discrete divergence operation.
\im {\bf Harmonic:}
    Any signal $\bbs \in \operatorname{kernel}(\bbB_{k+1}^\top) \cap \operatorname{kernel}(\bbB_k)$ is curl-free and divergence-free and therefore does not contribute to either circulation around a $(k+1)$-simplex or net flow through a $(k-1)$-simplex.
\end{itemize}
Thus, we may write the gradient and curl components as 
\alna{
    \bbs_{\rm G}^k = \bbB_k^\top \bbx^{k-1},
    \qquad
    \bbs_{\rm C}^k = \bbB_{k+1} \bbx^{k+1}
\nonumber}
for some $(k-1)$-signal $\bbx^{k-1}$ and $(k+1)$-signal $\bbx^{k+1}$.
Thus, $\bbs_{\rm G}^k$ is driven by lower-level connections, that is, $(k-1)$-faces and $\bbs_{\rm C}^k$ by upper-level connections, or $(k+1)$-cofaces, while the harmonic component $\bbs_{\rm H}^k$ represents the remaining portion of $\bbs^k$ that is orthogonal to both.

Because each of the three components in~\eqref{eq:hodge_components} correspond to mutually orthogonal subspaces, we also enjoy a spectral characterization of signals based on the eigendecomposition of the Hodge Laplacian $\bbL_k$.
We let $\bbU_k$ denote the eigenvectors of $\bbL_k$, which we further decompose as $\bbU_k = [ \bbU_{k, {\rm G}}, \bbU_{k, {\rm C}}, \bbU_{k, {\rm H}} ]$, where $\bbU_{k,{\rm G}}$ and $\bbU_{k, {\rm C}}$ denote the eigenvectors associated with nonzero eigenvalues of $\bbL_{k,\ell}$ and $\bbL_{k,u}$, respectively, while $\bbU_{k,{\rm H}}$ span the nullspace of $\bbL_k$.
We can therefore equivalently write the $k$-signal $\bbs^k$ as
\alna{
    \bbs^k = 
    \bbU_{k, {\rm G}} \tbs^k_{\rm G}
    +
    \bbU_{k, {\rm C}} \tbs^k_{\rm C}
    +
    \bbU_{k, {\rm H}} \tbs^k_{\rm H},
\label{eq:hodge_tft}}
where the spectral components $\tbs^k_{\rm G} :=  \bbU_{k, {\rm G}}^\top \bbs^k$,
$\tbs^k_{\rm C} :=  \bbU_{k, {\rm C}}^\top \bbs^k$,
and
$\tbs^k_{\rm H} :=  \bbU_{k, {\rm H}}^\top \bbs^k$
comprise $\tbs^k := \bbU_k^\top \bbs^k$, which contains the topological frequency coefficients.
More specifically, $\bbU_k^\top \bbs^k$ yields the topological Fourier transform (TFT) of $\bbs^k$, which allows us to generalize frequency decompositions for signals on simplicial complexes~\cite{barbarossa2020topological, schaub2021signal}.

\subsection{Dirac Decomposition}\label{Ss:dirac_decomp}

The Hodge decomposition is a classical, well-founded approach to model the spectral structure of simplicial signals for a single, given order.
We can further apply the Hodge Laplacians to build operators to model signal interactions across different orders.
When we are interested in \emph{multiorder} signals, that is, simplicial signals observed over all orders simultaneously, we resort to the \emph{Dirac operator} $\bbD$~\cite{calmon2022higher,baccini2022weighted}, which
must be of size $N \times N$ to account for the $N = \sum_{k=1}^K N_k$ simplices across all orders of the simplicial complex $\ccalP^K$.
Additionally, $\bbD$ has a $K\times K$ block structure, where the $(k,k')$-th block, denoted by $\llbracket \bbD \rrbracket_{k,k'}$, relates orders $k$ and $k'$.
Formally, each block is defined as~\cite{calmon2022higher,baccini2022weighted}
\begin{equation}\label{eq:dirac}
    \llbracket{\bbD}\rrbracket_{k,k'}
    =\left\{
    \begin{array}{cc}
    \mathbf{B}_k & \text{if}\;k'=k+1\\
    \mathbf{B}_{k'}^\top & \text{if}\;k'=k-1\\
    \mathbf{0} & \text{otherwise}
    \end{array}
    \right\}.
\end{equation}
Observe that, since we only consider interactions between adjacent orders, most blocks of $\bbD$ are zero-valued except the super- and sub-diagonal blocks.
Applying $\bbD$ as an operator $\bby = \bbD\bbx$ maps the simplicial complex signal space to itself.
However, if we structure $\bbx$ and $\bby$ into block vectors conformal with $\bbD$, we may write the $k$-th subvector of $\bby$ as $\llbracket \bby \rrbracket_{k} = \bbB_k \llbracket\bbx\rrbracket_{k-1} + \bbB_{k+1}^\top \llbracket \bbx\rrbracket_{k+1}$, that is, $\bbD$ maps signal values to the $k$-signal $\llbracket \bby \rrbracket_{k}$ from adjacent dimensions $k\pm 1$.

An equally important property of $\bbD$ regards its spectral characteristics.
The square of the Dirac operator yields a block-diagonal matrix $\ccalL = \bbD^2$ comprising the Hodge Laplacians along the diagonal, that is, $\ccalL := \operatorname{blkdiag}( \{ \bbL_k \}_{k=0}^K )$.
Indeed, the operator $\bbD$ maps signals of orders $k\pm 1$ to $k$-signals, so two applications $\bbD^2$ should map $k$-signals to $k$-signals through both upper- and lower-level connections, that is, the upper and lower Laplacians in~\eqref{eq:hodge_lap}.
Recalling the orthogonality of $\operatorname{span}(\bbB_k^\top)$ and $\operatorname{span}(\bbB_{k+1})$ and the definition in~\eqref{eq:hodge_lap}, we can relate the eigenvectors of $\bbD$ to the singular vectors of the incidence matrices.
Moreover, with any even power of $\bbD$, we can recover the spectral representations of the Hodge decomposition in Section~\ref{Ss:hodge_decomp}.
Thus, we can apply the TFT to multiorder signals through the eigenvectors of $\bbD$, analogous to the Hodge-based TFT in~\eqref{eq:hodge_tft}.

\section{Topological Weak Stationarity} \label{S:StatDefinition}

We now turn to the crux of this work, the definition of stationary random topological signals.
In classical time-series analysis, a stochastic process is said to be \emph{weakly stationary} if its correlation matrix is invariant to time shifts~\cite{hayes2009statistical, stoica2005spectral}.
This implies two fundamental properties for time-stationary processes.
First, a signal is weakly time-stationary if and only if it can be written as the output of a linear time-invariant filter given white noise.
Second, the spectral properties of the correlation matrix of any time-stationary process completely describe its second-order statistics.
For our generalization to topological domains, we consider these two properties to define weak stationarity for simplicial complex signals.

However, unlike stationarity for the graph domain~\cite{EPFL16stationary, marques2017stationary}, we must also account for the multiorder nature of the domain.
Indeed, we may be interested in multiorder signals observed across all orders.
Even if we only required signals over a single order, the simplicial complex domain will reflect higher-order interactions among nodes, edges, triangles, and higher-dimensional simplices.
We therefore seek a definition of weak stationarity that is consistent with the spectral geometry induced by the Hodge Laplacians in Section~\ref{Ss:hodge_decomp} or the Dirac operator in~\ref{Ss:dirac_decomp}.
In particular, we define stationary topological random signals as those that can be modeled as the output of a \emph{topological filter} given white noise, where the filter is characterized by either $\bbL_k$ or $\bbD$, depending on our interest.

%
\begin{mydefinition}\label{D:WeaklyStationaryTopoProcess}
\normalfont
    For some simplicial complex $\ccalP^K$, let $\bbT \in \{ \bbL_k, \bbD \}$.
    A zero-mean random topological signal $\bbs$ is \emph{weakly stationary} with respect to $\bbT$ if it can be expressed as
    \begin{equation}\label{eq:topo_stationarity_def}
        \bbs = \bbH \bbw,
    \end{equation}
    where $\bbw$ is a white simplicial signal satisfying $\mbE[\bbw] = \bbzero$ and $\bbE[\bbw\bbw^\top] = \bbI$, and the \emph{topological filter} $\bbH$ commutes with $\bbT$.
\end{mydefinition}

%

Our definition of weak stationarity for simplicial signals is therefore characterized by linear shift-invariant filters.
However, our description of topological filters is more flexible than for the time or graph domain.
We consider linear filters that aggregate shifted versions of input signals over the underlying domain, where shift correspond to upper- or lower-level connections, that is, faces or cofaces
We may further choose between filtering a $k$-signal at a single order or a multiorder signal~\eqref{eq:simplicial_complexes_signal} on the joint simplicial space.
More specifically, let $\bbT \in \{ \bbL_k, \bbD \}$ denote a topological operator of interest.
We define a topological filter as a polynomial\footnote{While not presented here due to space limitations, generative filters that operate independently in the lower and upper Laplacians can also be considered~\cite{roddenberry2022hodgelets}.} of $\bbT$,
\alna{
    \bbH = \sum_{r = 0}^{R - 1} h_{r} \bbT^{r}
\label{eq:top_filter}}
for real-valued coefficients $\bbh := [ h_0,\dots,h_{R-1} ]^\top \in \reals^{R}$.
If $\bbT = \bbL_k$ for a given order $k$, we obtain a \emph{Hodge-domain} filter processing $k$-signals via upper- and lower-level connections.
Alternatively, for $\bbT = \bbD$, we process signals supported jointly on all simplicial orders, yielding a \emph{Dirac-domain} filter.
Furthermore, observe that both $\bbL_k$ and $\bbD$ are symmetric.
Thus, $\bbT$ is symmetric and therefore diagonalizable with eigendecomposition $\bbT = \bbU \diag(\bblambda) \bbU^\top$.
Then, the filter $\bbH$ has the spectral representation
\alna{
    \bbH = \bbU \diag\big( h(\bblambda) \big) \bbU^\top, 
    \quad
    h(\lambda_i) = \sum_{r=0}^{R-1} h_r \lambda_i^r,
\label{eq:top_filter_spectral}}
where $h(\cdot)$ is the \emph{topological frequency response} of the filter, which is applied elementwise to the eigenvalues $\bblambda$.

Recall that weak stationarity is associated with two properties, behavior as output of a filter given white noise and spectral properties of signal correlation.
Our Definition~\ref{D:WeaklyStationaryTopoProcess} is clearly inspired by the first property.
We therefore next introduce the covariance matrix of stationary topological signals to consider the second property.
By Definition~\ref{D:WeaklyStationaryTopoProcess}, we can express the covariance of a stationary topological signal $\bbs$ as
\alna{
    \bbC 
    &~:=~&
    \mbE[\bbs\bbs^\top]
    \,=\,
    \mbE[ (\bbH\bbw) (\bbH \bbw)^\top ]
    \,=\,
    \bbH \mbE[ \bbw\bbw^\top ] \bbH^\top
&\nonumber\\&
    &~=~&
    \bbH\bbH^\top
    \,=\,
    \bbH^2,
\label{eq:topo_covariance}}
where the last inequality is by the symmetry of $\bbH$, as $\bbT$ is symmetric.
Thus, the covariance is entirely determined by the filter $\bbH$, which itself depends on the choice of operator $\bbT$ and the filter coefficients $\bbh$.
Note that, since $\bbC = \bbH^2 = \bbU \diag^2\big(h(\bblambda)\big) \bbU^\top$, the rank of $\bbC$ is the same as that of $\bbH$.

This characterization of $\bbC$ leads to an almost equivalent and often more interpretable definition of weak stationarity, which we derive from the spectrum of the covariance $\bbC$.
In the time domain, the autocorrelation of a stationary process is diagonalized by the Fourier transform.
Analogously, the covariance of a stationary topological process is diagonalizable by the TFT via the eigenbasis of $\bbT \in \{ \bbL_k, \bbD \}$.
This is formalized in the following definition.

%
\begin{mydefinition}\label{D:TopoStationary_SimDiag}
\normalfont
A zero-mean random topological signal $\bbs$ is \emph{weakly stationary} with respect to $\bbT \in \{ \bbL_k, \bbD \}$ if the covariance $\bbC$ and $\bbT$ are \emph{simultaneously diagonalizable}, that is,
\alna{
    \bbC = \bbU \diag(\bbp) \bbU^\top,
\label{eq:topo_simdiag}}
where $\bbU$ denotes the shared eigenvectors of $\bbC$ and $\bbT$ and $\bbp \geq \bbzero$ contains the eigenvalues of $\bbC$, which we refer to as the \emph{power spectral density}.
\end{mydefinition}
Observe that Definition~\ref{D:TopoStationary_SimDiag} is slightly more general than Definition~\ref{D:WeaklyStationaryTopoProcess}.
To see this, recall that for any stationary signal $\bbs$ from~\eqref{eq:topo_stationarity_def}, its covariance matrix $\bbC$ takes the form in~\eqref{eq:topo_covariance}, which shares eigenvectors $\bbU$ with $\bbT$ and has eigenvalues $\diag^2\big( h(\bblambda) \big)$.
Thus, $\bbs$ also satisfies Definition~\ref{D:TopoStationary_SimDiag}.
However, there are no constraints on the eigenvalues of $\bbC$ in~\eqref{eq:topo_simdiag}, whereas Definition~\ref{D:WeaklyStationaryTopoProcess} implicitly requires that, for any two eigenvalues $\lambda_i$ and $\lambda_j$ of $\bbT$, if $\lambda_i = \lambda_j$, then $p_i = p_j$. 
Critically, Definition~\ref{D:TopoStationary_SimDiag} emphasizes the importance of the PSD $\bbp$ in characterizing the second-order statistics of stationary topological signals.
Given the eigenvectors of $\bbT$, the PSD completely describes the covariance matrix $\bbC$.
We formalize the definition of $\bbp$ as follows.

%
\begin{mydefinition}\label{def:topo_psd}
\normalfont
For a random topological signal $\bbs$ that is stationary with respect to $\bbT = \bbU \diag(\bblambda) \bbU^\top \in \{ \bbL_k, \bbD \}$, its \emph{power spectral density (PSD)} is defined as
\begin{equation}\label{eq:topo_psd}
    \bbp := \diag( \bbU^\top \bbC \bbU ).
\end{equation}
\end{mydefinition}

Given the eigenbasis $\bbU$, the vector $\bbp$ not only completely characterizes $\bbC$ but also serves as the direct counterpart of the classical PSD used for stationary time-series processes.
In fact, the PSD of stationary topological signals inherits the same key property as its time-domain counterpart, formally stated next.

\begin{corollary}\label{C:tft_cov}
Let $\bbU$ denote the eigenbasis of the underlying topological operator $\bbT \in \{ \bbL_k, \bbD \}$ for $k \in \{ 0,1,\cdots,K \}$ and consider the TFT of $\bbs$, that is, the spectral representation 
\alna{
    \tbs = \bbU^\top\bbs.
\label{eq:topo_freq_signal}}
Then, the covariance of $\tbs$ is diagonal and given by
\alna{
    \tbC = \mbE[\tbs\tbs^\top] = \diag(\bbp).
\label{eq:topo_freq_cov}}
\end{corollary}

\noindent\textbf{Proof.}
The result follows directly since
\alna{
    \mbE[\tbs\tbs^\top]
    =
    \bbU^\top \mbE[\tbs\tbs^\top] \bbU
    =
    \bbU^\top \bbC \bbU
    =
    \diag(\bbp).
&\nonumber\\[-1.23cm]&
\nonumber}
\hfill$\square$

\medskip

This corollary formalizes our claim that Definition~\ref{def:topo_psd} is analogous to the PSD of time-stationary processes.
More specifically, note that the covariance $\tbC$ is diagonal, indicating that stationary topological signals are \emph{uncorrelated in the spectral domain}, where the spectral representation of signals is defined by the eigenvectors of the topological operator $\bbT$.
This simplification already justifies the use of the spectral domain for signal analysis, but it also implies another general property of stationarity.
Because entries of $\bbp$ are uncorrelated, we need only consider the $i$-th entry $p_i$ to completely represent the power of the stationary process at the $i$-th topological frequency.
Thus, the PSD $\bbp$ in~\eqref{eq:topo_simdiag} can be seen as the distribution of energy of the random signal $\bbs$ across the topological frequencies defined by $\bbT$.

Furthermore, under the conditions of Definition~\ref{D:WeaklyStationaryTopoProcess}, we can generalize the classical relationship between linear shift-invariant filters and the PSD of stationary processes in the time and graph domain to the simplicial complex domain.
In particular, given a simplicial signal $\bbs = \bbH\bbw$ as in~\eqref{eq:topo_stationarity_def} and the covariance in~\eqref{eq:topo_covariance}, we can write $\bbp$ as
\alna{
    p_i = h(\lambda_i)^2
    \quad
    \forall ~ i \in [N] := \{1,\dots,N\},
\label{eq:topo_psd_filter_relation}}
therefore the PSD corresponds to the squared magnitude of the filter frequency response, as expected.

By~\eqref{eq:topo_psd} and~\eqref{eq:topo_psd_filter_relation}, we observe another advantage of stationarity in terms of probabilistic modeling. 
Given the topological operator $\bbT = \bbU \diag(\bblambda)\bbU^\top$, we require only $N$ degrees of freedom to describe the covariance $\bbC$ of the stationary signal $\bbs$ rather than $N (N+1) / 2$.
Moreover, if $R < N$ for filter coefficients $\bbh \in \reals^R$, we may further reduce the degrees of freedom by modeling the PSD as in~\eqref{eq:topo_psd_filter_relation}.
Indeed, such a reduction is crucial, particularly as $N (N+1) / 2$ can grow very large for high-dimensional topological signals, which we exploit in the next section for covariance estimation.

\section{Covariance Estimation} \label{S:estimators}

To describe the statistical relationships between simplices, we consider the task of estimating the covariance matrix $\bbC = \mbE[\bbs\bbs^\top]$ of a stationary topological signal $\bbs$ from a set of $M$ realizations $\bbS = [\bbs_1, \dots, \bbs_M] \in \reals^{N \times M}$.
The typical approach to recovering $\bbC$ is the sample estimate
\alna{
    \hbC 
    &~:=~&
    \frac{1}{M} \bbS \bbS^\top
    ~=~
    \frac{1}{M}
    \sum_{m=1}^M
    \bbs_m \bbs_m^\top,
\label{eq:sample_cov}}
which requires $\ccalO( MN^2 )$ operations to compute.
While~\eqref{eq:sample_cov} is straightforward, it does not exploit known structure in $\bbC$ for stationary signals introduced in Section~\ref{S:StatDefinition}.
Thus, we aim to exploit characterizations in~\eqref{eq:topo_psd} and~\eqref{eq:topo_psd_filter_relation} to develop more efficient estimators of both the PSD $\bbp$ and the covariance $\bbC$.
Recall that we may consider signals at a single order $k$ via the Hodge Laplacian $\bbL_k$ or a multiorder signal over all orders via the Dirac operator $\bbD$.
Thus, for the following estimators, we represent both scenarios with $\bbT \in \{ \bbL_k, \bbD \}$ with eigendecomposition $\bbT = \bbU\diag(\bblambda)\bbU^\top$, where $\bbU^\top$ further induces the TFT.

\subsection{Nonparametric Estimators}\label{Ss:nonparam_estimators}

We introduce two nonparametric approaches, one for which we apply signals in the spatial domain (correlogram-based) and another for which we consider spectral representations of signals (periodogram-based).

\subsubsection{Correlogram-based estimator}

Our first approach consists of two steps.
First, we estimate the PSD by transforming the sample covariance $\hbC$ to the spectral domain,
\alna{
    \hbp_{\rm cg}
    :=
    \diag( \bbU^\top \hbC \bbU ) \in \reals_+^N.
\label{eq:est_psd_cg}}
Our covariance estimate is then reconstructed from $\hbp_{\rm cg}$ as
\alna{
    \hbC_{\rm cg} := \bbU \diag(\hbp_{\rm cg}) \bbU^\top.
\label{eq:est_cov_cg}}
Thus, $\hbC_{\rm cg}$ can be seen as the projection of $\hbC$ onto the subspace of symmetric matrices that are diagonalized by $\bbU$.
In terms of computational complexity, note that we must obtain $\hbC$ for $\ccalO(MN^2)$, followed by two spectral-domain transformations, each of which costs $\ccalO( N^3 )$.

\subsubsection{Periodogram-based estimator}

Alternatively, we apply the TFT for a computationally lighter approach, where we first transform the signals to the spectral domain,
\alna{
    \tbS
    &~=~&
    \bbU^\top \bbS
    ~=~
    [ \tbs_1, \dots, \tbs_M ].
\nonumber}
We then estimate the PSD by averaging the squared magnitude of the transformed signals,
\alna{
[\hbp_{\rm pg}]_i
:=
\frac{1}{M}
\sum_{m=1}^M
[\tbs_m]_i^2
\quad
\forall ~
m \in [M],
\label{eq:est_psd_pg}}
or equivalently, $\hbp_{\rm pg} = \frac{1}{M} (\tbS \circ \tbS) \bbone$, where $\circ$ denotes the Hadamard (elementwise) product.
The covariance estimate is obtained as before,
\alna{
    \hbC_{\rm pg}
    &~:=~&
    \bbU \diag(\hbp_{\rm pg}) \bbU^\top.
\label{eq:est_cov_pg}}
Observe that obtaining $\hbp_{\rm cg}$ in~\eqref{eq:est_psd_cg} requires $\ccalO( \max(M,N) N^2 )$, while $\hbp_{\rm pg}$ in~\eqref{eq:est_psd_pg} requires $\ccalO( MN^2 )$. 
However, both approaches~\eqref{eq:est_psd_cg} and~\eqref{eq:est_psd_pg} yield equivalent estimates, which is straightforward to verify.
Therefore, in low-sample regimes with $M < N$, estimating the PSD via~\eqref{eq:est_psd_pg} is computationally preferable.

We next validate the value of our proposed, numerically equivalent approaches in~\eqref{eq:est_psd_cg} and~\eqref{eq:est_psd_pg} by analyzing their mean and variance~\cite{marques2017stationary}.
Since they are equivalent, we proceed by analyzing $\hbp_{\rm pg}$, but the following discussion also holds for $\hbp_{\rm cg}$.
In particular, the following proposition shows that $\hbp_{\rm pg}$ is an unbiased estimate with uncorrelated entries, the variances of which depend on the magnitude of entries in $\bbp$.

%
\begin{myproposition}\label{P:CovariancePeriodogram}\normalfont
Consider a topological signal $\bbs$ that is stationary with respect to $\bbT \in \{ \bbL_k,\bbD \}$ and has PSD $\bbp$.
Given the matrix $\bbS \in \reals^{N \times M}$ of $M$ independent realizations of $\bbs$, 
let $\hbp_{\rm pg}$ be the periodogram-based PSD estimate in~\eqref{eq:est_psd_pg}.
Then,
\begin{itemize}[left= 8pt .. 16pt, noitemsep]
\im[(i)] 
The bias $\bbb_{\rm pg}$ of $\hbp_{\rm pg}$ is given by
    \alna{
        \bbb_{\rm pg}
        &~:=~&
        \mbE[ \hbp_{\rm pg} ] - \bbp
        ~=~
        \bbzero;
    \label{E:BiasPeriodogram}}
\im[(ii)] 
If $\bbs$ is Gaussian, the covariance $\bbSigma_{\rm pg}$ of $\hbp_{\rm pg}$ is
    \alna{
        \bbSigma_{\rm pg}
        &~:=~&
        \mbE\big[ (\hbp_{\rm pg} - \bbp) 
        (\hbp_{\rm pg} - \bbp)^\top \big]
        \,=\,
        \frac{2}{M} 
        \diag^2(\bbp).
    \label{E:CovariancePeriodogram}}
\end{itemize}
\end{myproposition}

%
\begin{proof}
The proof follows the same steps as that in \cite[Prop.~2]{marques2017stationary}, replacing the graph Fourier transform with the TFT.
\end{proof}

%

Thus, we have shown that $\hbp_{\rm pg}$ is an \emph{unbiased} estimate of $\bbp$, that is, $\bbb_{\rm pg} = \bbzero$.
Furthermore, if the signal $\bbs$ is Gaussian, we can guarantee that the estimate $\hbp_{\rm pg}$ has \emph{uncorrelated entries}.
The restriction to Gaussianity for the covariance of $\hbp_{\rm pg}$ is analogous to stationarity in the time and graph settings, which is due to the computation of fourth-order moments of $\bbs$~\cite[Sec.~8.2]{hayes2009statistical}.
Beyond merely being uncorrelated, we further note that the variance of $[\hbp_{\rm pg}]_i$ is proportional to $p_i^2$.
Thus, since $\hbp_{\rm pg}$ is an unbiased estimate, its expected error is dictated by its variance, and therefore by the magnitude of $\bbp$.
More specifically, by the definition of $\bbSigma_{\rm pg}$, we have that
\alna{
    \MSE_{\rm pg}
    &\,=\,&
    \mbE\big[ \| \hbp_{\rm pg} - \bbp \|_2^2 \big]
    \,=\,
    \mbE\big[
        \tr\big( 
        (\hbp_{\rm pg} - \bbp) 
        (\hbp_{\rm pg} - \bbp)^\top
        \big)
    \big]
&\nonumber\\&
    &\,=\,&
    \tr(\hbC_{\rm pg})
    \,=\,
    \frac{2}{M} \| \bbp \|_2^2.
\label{E:mse_period}}
Unsurprisingly, the MSE of $\hbp_{\rm pg}$ is inversely proportional to the number of observations $M$.
In scenarios with low $M$, we may consider alternative techniques to reduce $\MSE_{\rm pg}$ at the cost of potential bias, which we leave for future exploration.
Note that by~\eqref{eq:topo_simdiag} and~\eqref{eq:est_cov_pg}, the true covariance $\bbC$ and estimate $\hbC_{\rm pg}$ come from unitary transformations of $\diag(\bbp)$ and $\diag(\hbp_{\rm pg})$, respectively, and therefore $\hbC_{\rm pg}$ is similarly an unbiased estimate of $\bbC$ with MSE $\frac{2}{M} \| \bbp \|_2^2 = \frac{2}{M} \| \bbC \|_F^2$.

\subsubsection{Subspace-aware estimation}\label{S:subspace_psd}

Real-world topological signals are often restricted to a particular subspace.
For example, in~\eqref{eq:hodge_components}, edge flow signals that represent communication traffic are typically dominated by the harmonic component, lying primarily in the $\bbU_{k,{\rm H}}$ subspace. 
In such cases, we can exploit the restricted support of signals in the spectral domain as additional structure when estimating the covariance $\bbC$.

Let $\bbU_{\ccalS}$ denote the columns of the eigenvectors $\bbU$ corresponding to the subspace of interest, indexed by $\ccalS \subseteq [N]$.
When $\bbU_{\ccalS}$ is known, we can rewrite the covariance as $\bbC = \bbU \diag(\bbp) \bbU^\top = \bbU_{\ccalS} \diag( \bbp_{\ccalS} ) \bbU_{\ccalS}^\top$, where $\bbp_{\ccalS}$ contains the nonzero values of the PSD.
Then, the nonzero entries of our estimators in~\eqref{eq:est_psd_cg} and~\eqref{eq:est_psd_pg} become
\alna{
    \!
    [\hbp_{{\rm cg}}]_{\ccalS}
    &\,=\,&
    \diag( \bbU_{\ccalS}\hbC\bbU_{\ccalS}^\top ),
    ~~
    [\hbp_{{\rm pg}}]_{\ccalS}
    \!=\!
    \frac{1}{M}
    \big( (\bbU_{\ccalS}^\top \bbS) \circ (\bbU_{\ccalS}^\top \bbS) \big)\bbone,
    \!
\nonumber}
with the subsequent covariance estimation proceeding as in~\eqref{eq:est_cov_cg} and~\eqref{eq:est_cov_pg}, with $\bbU_{\ccalS}$ in place of $\bbU$.

In some cases, we may know that the signal $\bbs$ is confined to a subspace that is \emph{unknown}.
For example, consider a $k$-signal $\bbs$ and $\bbT = \bbL_k = \bbU_{k}\diag(\bblambda_k)\bbU_k^\top$, which is restricted to the gradient, curl, or harmonic subspace, that is, $\bbU_{\ccalS} \in \{ \bbU_{k,{\rm G}}, \bbU_{k,{\rm C}}, \bbU_{k,{\rm H}} \}$ as defined in~\eqref{eq:hodge_components}, but we do not know which.
We can infer which subspace is active before estimating the PSD $\bbp$ or covariance $\bbC$.
As a simple yet effective strategy, we project $\bbs$ onto each of the three mutually orthogonal subspaces and compute the energy for each,
\alna{
    E_{\rm G} = \| \bbU_{k, {\rm G}}^\top \bbs \|_2^2,
    ~
    E_{\rm C} = \| \bbU_{k, {\rm C}}^\top \bbs \|_2^2,
    ~
    E_{\rm H} = \| \bbU_{k, {\rm H}}^\top \bbs \|_2^2.
\label{eq:subspace_detection}}
We can then identify the most likely component through energy thresholds or likelihood-based tests~\cite{liu2025matched}, followed by PSD and covariance estimation as before.

\subsection{Parametric Estimators}

We next introduce topological filter models for simplicial signals that allow \emph{parametric} approaches to estimate $\bbp$ and $\bbC$.
In particular, we define generalizations of moving-average (MA) and AR stochastic processes for the simplicial domain.
These models correspond to stationary processes, and, as we will show, this reduces complexity of modeling the covariance $\bbC$ from $\ccalO(N^2)$ to $\ccalO(R)$, the number of filter parameters. 

\subsubsection{MA Topological Processes} 
\label{S:MA_topo}

As in Definition~\ref{D:WeaklyStationaryTopoProcess}, we assume that our stationary topological signals are generated as the output of a finite-order topological filter in~\eqref{eq:top_filter} driven by white noise.
As mentioned in Section~\ref{S:StatDefinition}, this leads to parametric covariance estimators via the model in~\eqref{eq:topo_psd_filter_relation} since we need only recover the filter parameters.

This allows us to define \emph{MA topological processes}, that is, signals following~\eqref{eq:topo_stationarity_def}.
To formalize the parametric model for MA processes, we let $\bbbeta = [\beta_0, \dots, \beta_{R-1}]^\top$ denote the filter coefficients and $\bbT \in \{ \bbL_k, \bbD \}$ our symmetric topological operator of interest.
Then, our finite-impulse-response (FIR) topological filter is defined as
\alna{
    \bbH(\bbbeta)
    ~=~
    \sum_{r = 0}^{R-1}
    \beta_r \bbT^r,
\label{eq:topo_FIR_filter}}
where $R < N$, and, in practice, $R \ll N$.
Since our stationary signal $\bbs$ is generated from white noise $\bbw$, we have
\alna{
    \bbs = \bbH(\bbbeta) \bbw,
    \quad
    \mbE[\bbw\bbw^\top] = \bbI,
\label{eq:signal_MA}}
with covariance
\alna{
    \bbC^{\beta}(\bbbeta)
    =
    \bbH(\bbbeta)\bbH(\bbbeta)^\top
    =
    \sum_{r = 0}^{R-1}
    \sum_{r' = 0}^{R-1}
    \beta_r \beta_{r'}
    \bbT^{r + r'},
\label{eq:cov_topo_MA}}
which is itself a polynomial of degree $2(R-1)$ in $\bbT$
\begin{equation}\label{eq:cov_gamma_topo}
    \bbC^{\gamma}(\bbgamma) = \sum_{r=0}^{2(R-1)} \gamma_r \bbT^r,
    \qquad
    \gamma_r := \sum_{r'+r''=r} \beta_{r'}\beta_{r''},
\end{equation}
with coefficients $\bbgamma \in \reals^{2(R-1)}$ obtained via convolution.

Then, with the eigendecomposition $\bbT = \bbU \diag(\bblambda) \bbU^\top$ and the definition in~\eqref{eq:topo_psd}, the PSD corresponding to $\bbC^{\beta}(\bbbeta)$ is
\alna{
    \bbp(\bbbeta)
    &~=~&
    h^{\beta}(\bblambda)
    \circ
    h^{\beta}(\bblambda)
    =
    (\bbPsi^{\beta} \bbbeta)
    \circ
    (\bbPsi^{\beta} \bbbeta),
\label{eq:topo_psd_beta}}
where $h^{\beta}(\lambda) = \sum_{r=0}^{R-1} \beta_r \lambda^r$ denotes the topological frequency response of $\bbH^\beta(\bbbeta)$, and $\bbPsi^{\beta} \in \reals^{N \times R}$ is a Vandermonde matrix such that $\Psi^{\beta}_{ir} = \lambda_i^{r-1}$ for all $i \in [N], r \in [R]$.
Thus, our definition of MA topological processes mirrors that of the time and graph domains, where the PSD is equal to the filter frequency response squared.

If $\bbs$ is an MA process, we can estimate the covariance $\bbC$ by recovering $\bbbeta$ given the sample covariance $\hbC$ in~\eqref{eq:sample_cov}, 
\alna{
    \hbbeta_{\bbC}
    &~=~&
    \argmin\nolimits_{\bbbeta}
    \ccalD_{\bbC}
    \Big(
        \hbC, \bbC^{\beta}(\bbbeta)
    \Big),
\label{eq:opt_beta_space_topo}}
with $\bbC^{\beta}(\bbbeta)$ in~\eqref{eq:cov_topo_MA} and $\ccalD_{\bbC}(\cdot,\cdot)$ as a distance metric quantifying covariance error.
As a potentially more efficient alternative, we may estimate $\bbbeta$ in the frequency domain given an estimate of the PSD, such as the periodogram-based $\hbp_{\rm pg}$ in~\eqref{eq:est_psd_pg},
\alna{
    \hbbeta_{\bbp}
    &~=~&
    \argmin\nolimits_{\bbbeta}
    \ccalD_{\bbp}
    \big(
        \hbp_{\rm pg}, \,
        ( \bbPsi^{\beta} \bbbeta \circ \bbPsi^{\beta} \bbbeta )
    \big),
\label{eq:opt_beta_freq_topo}}
with the metric $\ccalD_{\bbp}(\cdot,\cdot)$ measuring PSD error.

Regardless of our choice of~\eqref{eq:opt_beta_space_topo} or~\eqref{eq:opt_beta_freq_topo}, both objective functions are quadratic in $\bbbeta$ and therefore are typically nonconvex.
However, under additional assumptions on $\bbT$ or $\bbbeta$, we may relax~\eqref{eq:opt_beta_freq_topo} to be convex or tractable, such as for a positive semidefinite $\bbT$ and nonnegative $\bbbeta$.
Moreover, in the specific case when $\ccalD_{\bbp}(\bbp_1,\bbp_2) = \| \bbp_1 - \bbp_2 \|_2^2$, we can employ efficient phase-retrieval algorithms such as Wirtinger flow~\cite{candes2015phase,shechtman2015phase} to solve~\eqref{eq:opt_beta_freq_topo}.
We highlight one convex relaxation of~\eqref{eq:opt_beta_space_topo}, where we ignore the convolutional structure in $\gamma$ defined in~\eqref{eq:cov_gamma_topo}.
To be specific, we estimate $\bbC^{\gamma}(\bbgamma)$ instead of $\bbC^{\beta}(\bbbeta)$ for
\alna{
    \hbgamma_{\bbC}
    &~=~&
    \argmin\nolimits_{\bbgamma}
    \ccalD_{\bbC}
    \Big(
        \hbC, \bbC^{\gamma}(\bbgamma)
    \Big),
\label{eq:opt_gamma_topo}}
which is convex whenever the metric $\ccalD_{\bbC}(\cdot,\cdot)$ is convex.
Similarly, with $\bbPsi^{\gamma} \in \reals^{N \times 2R}$ such that $\Psi^{\gamma}_{ir} = \lambda_i^{r-1}$ for $i \in [N], r \in [2R-1]$, we can estimate $\bbgamma$ in the frequency domain via
\alna{
    \hbgamma_{\bbp}
    &~=~&
    \argmin\nolimits_{\bbgamma}
    \ccalD_{\bbp}
    \big(
        \hbp_{\rm pg},
        \bbPsi^{\gamma} \bbgamma
    \big).
\label{eq:opt_gamma_freq_topo}}
Observe that if $\ccalD_{\bbC}(\bbC_1,\bbC_2) = \| \bbC_1 - \bbC_2 \|_F^2$ and $\ccalD_{\bbp}(\bbp_1,\bbp_2) = \| \bbp_1 - \bbp_2 \|_2^2$, then the objective functions and therefore the solutions of~\eqref{eq:opt_gamma_topo} and~\eqref{eq:opt_gamma_freq_topo} are equivalent.
However, obtaining the closed-form solutions differ in complexity, which we show empirically in Section~\ref{S:exps}.
Thus, our generalization of the classical parameterized MA process for simplicial complexes simplifies the representation of the covariance matrix $\bbC$ from $\ccalO(N^2)$ degrees of freedom to $\ccalO(R)$ or $\ccalO(R^2)$, mitigating the complexity of recovering $\bbC$.

\subsubsection{AR Topological Processes}
\label{S:AR_topo}

We next introduce the AR filter model for topological signals, which describes the signal $\bbs$ via recursive filtering given white noise,
\alna{
    \bbs
    &~=~&
    \sum_{r = 1}^R \alpha_r \bbT^r \bbs + \bbw,
\label{eq:topo_AR_process}}
where $\bbw$ is a white simplicial signal with $\mbE[\bbw] = \bbzero$ and $\mbE[\bbw\bbw^\top] = \bbI$, and $\bbalpha = [ \alpha_1, \dots, \alpha_R ]^\top \in \reals^R$ contains the AR coefficients.
We can equivalently write the recursive process in~\eqref{eq:topo_AR_process} as a filtering operation, where we have the filter
\alna{
    \bbH(\bbalpha)\bbs = \bbw,
    \quad
    \bbH(\bbalpha) := \bbI - \sum_{r=1}^R \alpha_r \bbT^r,
\label{eq:topo_AR_filter_form}}
leading to the generative form of $\bbs$
\alna{
    \bbs = \bbH^{-1}(\bbalpha) \bbw.
\label{eq:topo_AR_gen}}
Then, the covariance matrix of $\bbs$ is
\alna{
    \bbC^{\alpha}(\bbalpha)
    =
    \bbH^{-2}(\bbalpha),
\label{eq:cov_topo_AR}}
while the PSD can be defined via the frequency response of $\bbH(\bbalpha)$, that is, 
\alna{
    [\bbp(\bbalpha)]_i
    &\,:=\,&
    \frac{1}{\big( h^{\alpha}(\lambda_i) \big)^2}
    \,=\,
    \frac{1}{[ \bbone - \bbPsi^{\alpha}\bbalpha ]_i^{2}}
    ~
    \forall \, i\in[N]
\label{eq:psd_topo_AR}}
for $h^\alpha(\lambda) = 1 - \sum_{r=1}^R \alpha_r \lambda^r$ and Vandermonde matrix $\bbPsi^{\alpha} \in \reals^{N \times R}$ with $\Psi^{\alpha}_{ir} = \lambda_i^r$ for all $i\in[N], r\in[R]$, where the topological operator again satisfies $\bbT = \bbU \diag(\bblambda) \bbU^\top$.

To estimate the covariance matrix of an AR topological process, we recover the filter parameters $\bbalpha$ by minimizing the discrepancy between $\bbC^{\alpha}(\bbalpha)$ and the sample covariance $\hbC$ in~\eqref{eq:sample_cov}.
However, to avoid requiring the inversion of $\bbH(\bbalpha)$ in~\eqref{eq:cov_topo_AR}, rather than comparing the two matrices $\bbC^{\alpha}(\bbalpha)$ and $\hbC$ directly, we instead compare $\hbC \big( \bbC^{\alpha}(\bbalpha) \big)^{-1}$ to the identity matrix, resulting in the following optimization problem
\alna{
    \hbalpha_{\bbC}
    &~=~&
    \argmin\nolimits_{\bbalpha}
    \ccalD_{\bbC}
    \Big(
        \hbC \big( \bbC^{\alpha}(\bbalpha) \big)^{-1}, \bbI
    \Big).
\label{eq:opt_alpha_topo}}
Analogously to the MA case, we can also estimate the filter parameters $\bbalpha$ in the frequency domain via the PSD~\eqref{eq:psd_topo_AR}, 
\alna{
    \hbalpha_{\bbp}
    &~=~&
    \argmin\nolimits_{\bbalpha}
    \ccalD_{\bbp}
    \Big(
        \diag^{-1}\big(\bbp(\bbalpha)\big) \hbp_{\rm pg},
        \bbone
    \Big),
\label{eq:opt_alpha_freq_topo}}
where $\hbp_{\rm pg}$ denotes the periodogram-based PSD estimate in~\eqref{eq:est_psd_pg}.
As before, the problems~\eqref{eq:opt_alpha_topo} and~\eqref{eq:opt_alpha_freq_topo} are quadratic in $\bbalpha$ and therefore may be nonconvex for common choices of $\ccalD_{\bbC}$ and $\ccalD_{\bbp}$.
Thus, we can again consider estimating an alternative set of parameters $\bbeta \in \reals^{ 2R }$ defined by the convolution of $\bbalpha$ with itself, where
\alna{
    \bbC^{\eta}(\bbeta)
    &~=~&
    \bigg(
    \bbI - 
    \sum_{r=1}^{2R} \eta_r \bbT^r
    \bigg)^{-1} \! ,
    ~~
    \eta_r := \sum_{r' + r'' = r} \alpha_{r'} \alpha_{r''}.
\label{eq:cov_eta_topo}}
This leads to the simplified formulations [cf.~\eqref{eq:opt_alpha_topo} and \eqref{eq:opt_alpha_freq_topo}]
\begin{align}
    \hbeta_{\bbC}
    ~=~&\argmin\nolimits_{\bbeta}
    \ccalD_{\bbC}
    \Big(
        \hbC 
        \big( 
            \bbI - {\textstyle\sum_{r=1}^{2R}} \eta_r \bbT^r
        \big), \,
        \bbI
    \Big)\label{eq:opt_eta_topo}\\
    \hbeta_{\bbp}
    ~=~&\argmin\nolimits_{\bbeta}
    \ccalD_{\bbp}
    \Big(
        \hbp_{\rm pg}
        \circ
        \big(
            \bbone - \bbPsi^{\eta}\bbeta
        \big), \,
        \bbone
    \Big)
\label{eq:opt_eta_freq_topo}
\end{align}
where $\bbPsi^{\eta} \in \reals^{N \times 2R}$ is the Vandermonde matrix with $\Psi^{\eta}_{ir} = \lambda_i^r$ for $i \in [N], r \in [2R]$.
While~\eqref{eq:opt_eta_freq_topo} is computationally simpler than~\eqref{eq:opt_eta_topo} as for the MA model, note that $\ccalD_{\bbC}(\bbC_1,\bbC_2) = \| \bbC_1 - \bbC_2 \|_F^2$ and $\ccalD_{\bbp}(\bbp_1,\bbp_2) = \| \bbp_1 - \bbp_2 \|_2^2$ do not yield equivalent objective functions due to additional terms in~\eqref{eq:opt_eta_topo} from off-diagonal matrix entries.
Since~\eqref{eq:opt_eta_topo} requires optimizing $N(N+1)/2$ terms instead of only the $N$ terms of~\eqref{eq:opt_eta_freq_topo}, we require more signals $M$ for a competitive solution from~\eqref{eq:opt_eta_topo}.
Thus, in scenarios where our estimate of $\hbp_{\rm pg}$ is more  reliable than $\hbC$, estimating $\hbeta_{\bbp}$ will likely be superior to $\hbeta_{\bbC}$ in both computation and sample complexity.

One interesting case is first-order AR processes ($R=1$),
\alna{
    \bbC = (\bbI - \alpha \bbT)^{-2},
\nonumber}
which enjoys a few advantages in estimating simplicial complex properties.
First, since the filter in~\eqref{eq:topo_AR_filter_form} requires only a scalar parameter $\alpha \in \reals$, our estimation of $\bbC$ in~\eqref{eq:opt_alpha_topo} and $\bbp$ in~\eqref{eq:opt_alpha_freq_topo} becomes computationally simpler given a known topological operator $\bbT \in \{ \bbL_k, \bbD \}$.
Second, the sparsity pattern of $\bbT$ directly affects the support of the precision matrix, $\bbC^{-1} = \bbI - 2\alpha\bbT + \alpha^2 \bbT^2$.
Thus, we may recover $\bbC$ by first estimating a sparse $\bbC^{-1}$.
Third, we can similarly leverage this relationship to estimate an unknown $\bbT$ by jointly estimating $\bbC^{-1}$ and $\bbT$, the practical and theoretical advantages of which we will explore in future work.

As a final note on these two signal models, we can combine the definitions of MA and AR topological processes for \emph{autoregressive–moving-average} (ARMA) signals, which take the form $\bbH_1(\bbalpha)\bbs = \bbH_2(\bbbeta) \bbw$ with polynomial operators $\bbH_1(\bbalpha)$ as in~\eqref{eq:topo_FIR_filter} and $\bbH_2(\bbbeta)$ as in~\eqref{eq:topo_AR_filter_form}.
While interesting, we leave the study of these ARMA models for future exploration.

\subsubsection{Non-Polynomial Filtering}
\label{S:BeyondPoly_topo}

We have thus far considered topological signal models based on polynomial filters of the operator $\bbT \in \{ \bbL_k, \bbD \}$, which are well-established models of diffused signals over simplicial complexes.
Here, we consider more general, non-polynomial filters as additional parametric representations of stationary topological processes.

\vspace{.1cm}
\noindent \emph{Spectral-domain filter parameterizations.}
Instead of defining a topological filter in the spatial domain, we consider some frequency response $h^{\theta}(\lambda)$ parameterized by a low-dimensional vector $\bbtheta$. 
As for MA and AR scenarios, we can estimate the covariance $\bbC$ via the filter frequency response, which requires that we recover $\bbtheta$ by
\alna{
    \hbtheta_{\rm spec}
    =
    \argmin\nolimits_{\bbtheta}
    \ccalD_{\bbp}\big(
        \hbp_{\rm pg},
        h^{\theta}(\bblambda)
        \circ
        h^{\theta}(\bblambda)
    \big).
\label{eq:opt_theta_filter}}
This formulation is quite general, but classical examples of such parametric models include exponential or rational responses such as
\alna{
    h^{\theta}(\lambda)
    =
    e^{-\theta_1 \lambda},
    \quad
    h^{\theta}(\lambda) = {\rm sigmoid}(\theta_1 \lambda + \theta_2),
\nonumber}
which act as classical heat diffusion or bandpass filters, respectively, for signals on simplicial complexes.

\vspace{.1cm}
\noindent \emph{Kernel parameterizations.}
Alternatively, we may parameterize the PSD $\bbp_{\theta} = \kappa^{\theta}(\bblambda)$ directly rather than the filter, that is, learning a spectral kernel, a function of the simplicial complex spectrum.
We then estimate the PSD via its parameters $\bbtheta$ as
\alna{
    \bbtheta_{\rm ker}
    &~=~&
    \argmin\nolimits_{\bbtheta}
    \ccalD_{\bbp}
    \big(
        \hbp_{\rm pg},
        \kappa^{\theta}(\bblambda)
    \big),
\label{eq:opt_theta_kernel}}
with examples including Gaussian or Laplacian kernels,
\alna{
    \kappa^{\theta}(\lambda) = e^{-\theta_1 \lambda^2},
    \quad
    \kappa^{\theta}(\lambda) = \frac{1}{( 1 + \theta_1 \lambda )^{\theta_2}}.
\nonumber}
These particular kernels can model smooth PSDs or power-law behaviors, which are observed in many real topological datasets~\cite{bick2023higher, salnikov2018simplicial, barabasi2013network}.
Observe that the difference between the spectral-based approach in~\eqref{eq:opt_theta_filter} and the kernel-based one in~\eqref{eq:opt_theta_kernel} is primarily based on interpretation.
Both minimize the difference between $\hbp_{\rm pg}$ and a parameterized function, the meaning of which determines our focus.

Not only do the spectral- and kernel-based filter models $h^{\theta}(\lambda)$ and $\kappa^{\theta}(\lambda)$ allow us more flexible low-dimensional parametrizations, but we can also relate them to topological machine learning.
For example, we can treat $\kappa^{\theta}(\lambda)$ as a learnable kernel to define a Gaussian process over the topological frequency domain to predict the PSD of any simplicial complex based on the spectrum $\bblambda$ of its associated topological operator $\bbT \in \{ \bbL_k, \bbD \}$.
More specifically, by~\eqref{eq:topo_freq_cov}, we can construct the covariance matrix of the frequency response $\tbs = \bbU^\top \bbs$ of any stationary topological process $\bbs$ as $\diag(\kappa^{\theta}(\bblambda))$.
Additionally, we can take a data-driven approach to learning the filter frequency response, where we can model $h^{\theta}(\lambda)$ as a neural network and fit the filter parameters $\bbtheta$ to observed topological signals.
For these examples, we can still exploit topological stationarity by learning filter parameters in the frequency domain, for which model complexity is $\ccalO(N)$ instead of $\ccalO(N^2)$, as before.

\begin{remark}[Single vs. Multiorder Models]\label{R:single_multiorder}
We focus on the topological operator $\bbT$ as either a Hodge Laplacian $\bbL_k$ for $k$-signals residing on a single order or the Dirac operator $\bbD$ for multiorder signals, with data observed at all dimensions at once.
For the former, we could still model signals at all orders by considering $K+1$ filters with coefficients $\{ \bbh_k \}_{k=0}^K$, one for each order, albeit with no relationship between filters of different dimensions.
For the latter, we stipulate that the diffusion dynamics across orders are equivalent, that is, the same filter coefficients $\bbh$ are shared across all dimensions.
However, we may also consider cases in between the two extremes.
For example, we could impose structural relationships across the coefficient sets $\{ \bbh_k \}_{k=0}^K$, such as smooth variations or hierarchical dependencies across orders.
While this allows greater modeling flexibility, we can still apply the proposed estimators in this section, merely modifying them to include any constraints on the filter coefficients.

\end{remark}

\section{Stationarity-based models and tools} \label{S:problem}

We next demonstrate the value of stationary topological processes for classical signal-processing tasks. We study topological denoising (generalizing the Wiener filter) and interpolation (showing that simplicial $k$-signals can be recovered without direct 
$k$-order observations) and examine Gaussian topological stationarity before concluding with open problems.

\subsection{Wiener Filtering for Topological Signals}
\label{S:Wiener_topo}

We begin with the task of recovering an unknown topological signal $\bbs$ that has been corrupted by additive noise, for which we generalize the classical Wiener filter for the simplicial complex domain. 
Similarly to covariance estimation, if $\bbs$ is stationary, then we can greatly improve the efficiency of the task by operating in the frequency domain, which we will show corresponds to a simple element-wise scaling operation.

We wish to recover a zero-mean topological signal $\bbs$ that is stationary with respect to $\bbT \in \{ \bbL_k, \bbD \}$ and has covariance $\bbC = \mbE[\bbs\bbs^\top]$.
In particular, given the corrupted signal
\alna{
    \bby = \bbs + \bbv,
\label{eq:Wiener_obs}}
where $\bbv$ is a zero-mean white random signal that is independent of $\bbs$ with covariance $\bbC_v = \sigma_v^2 \bbI$ for $\sigma_v^2 > 0$, we seek the minimum MSE (MMSE) estimate of $\bbs$ from $\bby$, that is,
\alna{
    \hbs\!=\!\argmin\nolimits_{\bbz}
    \mbE\big[\| \bbs - \bbz \|_2^2 | \bby\big]
    \!=\!\bbC ( \bbC + \bbC_v )^{-1} \bby=\bbG\bby,
\label{eq:Wiener_MMSE_problem}}
where the optimal solution comes from the definitions of $\bbs$, $\bby$, and topological stationarity, and $\bbG$ is the Wiener filter matrix~\cite{kay1998fundamentals,hayes2009statistical}.
While straightforward, it may be unwieldy and even infeasible to compute $\bbG$ for even medium-sized simplicial complexes, as it requires the inversion of an $N \times N$ matrix.
However, by topological stationarity, we can reduce the obtention of $\hbs$ dramatically.
Recall that $\bbC$ shares eigenvectors $\bbU$ with $\bbT$ since $\bbs$ is stationary.
Then, by Definition~\ref{def:topo_psd} and the whiteness of $\bbv$, we have that
\alna{
    \bbC = \bbU \diag(\bbp)\bbU^\top,
    \quad
    \bbC_v = \bbU \diag(\sigma_v^2 \bbone) \bbU^\top.
\nonumber}
Then, with signal TFTs $\tilde{\hbs} = \bbU^\top \bbs$ and $\tby = \bbU^\top \bby$, we can rewrite~\eqref{eq:Wiener_MMSE_problem} in the frequency domain as
\alna{
    \tilde{\hbs} 
    &~=~&
    \bbU^\top \bbG \bbU \tby
    \,=\,
    \diag(\bbp)
    \diag^{-1}
    \big(
        \bbp + \sigma_v^2 \bbone
    \big)
    \tby
    \,=\,
    \tbg \circ \tby
\label{eq:Wiener_spectral}}
where $\tbg$ denotes the frequency response of the topological Wiener filter $\bbG$, where
\alna{
    \tilde{g}_i
    &~:=~&
    \frac{p_i}{p_i + \sigma_v^2}
    ~~
    \forall ~ i\in[N].
\label{eq:Wiener_freq_response}}
Thus, obtaining the MMSE of $\bbs$ scales reasonably even for large simplicial complexes, as we need only perform element-wise scaling of the TFT $\tby$ with $\tbg$ for $\ccalO(N)$, followed by transforming back to the spatial domain $\hbs = \bbU \tilde{\hbs}$ for $\ccalO(N^2)$, as opposed to the $\ccalO(N^3)$ of~\eqref{eq:Wiener_MMSE_problem}.
Given~\eqref{eq:Wiener_freq_response}, the solution $\hbs$ preserves entries of $\tby$ with higher SNR, that is, $i \in [N]$ such that $p_i \gg \sigma_v^2$ ($g_i \approx 1$), while attenuating entries with lower SNR, where $p_i$ is small relative to $\sigma_v^2$ ($g_i \approx 0$).
Finally, note that we can replace the PSD $\bbp$ in~\eqref{eq:Wiener_freq_response} with an estimate such as $\hbp_{\rm pg}$ in~\eqref{eq:est_psd_pg} when $\bbp$ is unknown.

\subsection{Topological Signal Reconstruction and Interpolation}
\label{S:reconstruction_topo}

We next consider topological signal reconstruction under the assumption of stationarity.
In this case, we are given partially observed signals, such as measuring only a subset of nodes, edges, or higher-order simplices or missing entire simplicial orders, and we wish to recover the full, original signals.
As before, we exploit the additional covariance structure imposed by stationarity for statistically optimal signal interpolation.

As in Section~\ref{S:Wiener_topo}, we seek a zero-mean topological signal $\bbs$ that is stationary with respect to $\bbT \in \{ \bbL_k, \bbD \}$ and has covariance $\bbC$, but here, we observe only a subset of its entries.
More formally, our partial observations are modeled as
\alna{
    \bar{\bbs}
    &~=~&
    \bbTheta \bbs + \bbv,
\label{eq:interp_obs_model}}
where $\bbTheta \in \{ 0,1 \}^{P \times N}$ selects the $P < N$ observed entries of $\bbs$, and $\bbv$ again denotes zero-mean white noise independent of $\bbs$ with covariance $\bbC_v = \sigma_v^2 \bbI$.
As before, we first consider the optimal linear MMSE of $\bbs$ from~\eqref{eq:interp_obs_model}, which is given by
\alna{
    \hbs
\!=\!
\bbC\bbTheta^{\!\top}
\!\big(
\bbTheta\bbC\bbTheta^{\!\top}
\!+\!
\sigma_v^2 \bbI
\big)^{-1}
\bar{\bbs}\!=\!
    \big(
        \bbTheta^\top \bbTheta
        \!+\!
        \sigma_v^2 \bbC^{-1}
    \big)^{\!-1}
    \bbTheta^{\!\top }
    \bar{\bbs},
\label{eq:interp_closed_form}}
where $\bbTheta^\top\bbTheta$ is a binary diagonal matrix that encourages preserving the observed simplices $\bar{\bbs}$ and $\bbTheta^\top \bar{\bbs}$ is a zero-padded interpolation of the partially observed signal $\bar{\bbs}$.
However, note that~\eqref{eq:interp_closed_form} is also the solution to the quadratic problem
\alna{
    \hbs &~=~&
    \argmin\nolimits_{\bbz}
    \frac{1}{\sigma_v^2}
    \| \bar{\bbs} - \bbTheta \bbz \|_2^2
    +
    \bbz^\top \bbC^{-1} \bbz,
\label{eq:interp_map_form}}
so the optimal MMSE solution $\hbs$ can also be seen as obtained from a regularized optimization problem, where the term $\bbz^\top \bbC^{-1} \bbz$ promotes stationary behavior for the estimate $\hbs$.
Furthermore, if $\bbs$ is Gaussian, then $\bbz^\top \bbC^{-1} \bbz$ in~\eqref{eq:interp_map_form} acts as a prior and~\eqref{eq:interp_closed_form} is the maximum a posteriori (MAP) estimate.
Thus, the structure of $\bbC^{-1}$ influences our estimate of $\bbs$.
With known properties of the stationary signal $\bbs$ that appear in $\bbC$, we can specify the expression of the regularizer and possibly even impose additional structure to improve estimation.
We share a few particularly relevant choices below.

\begin{itemize}[left= 2pt .. 10pt, noitemsep]
    \im {\bf Smoothness:}
    We are often interested in signals that are smooth, that is, signals that do not vary greatly across connections.
    For stationary topological signals, we consider $\bbC^{-1} = \bbL_k$ if $\bbs$ is a $k$-signal that is smooth with respect to upper- and lower-level connections or $\bbC^{-1} = \bbD^2 = \ccalL$ if $\bbs$ is a multiorder signal that is smooth across all orders.
    The regularizer $\bbz^\top \bbC^{-1} \bbz$ in~\eqref{eq:interp_map_form} will penalize high signal variance across simplicial connections, prioritizing a smoother estimate $\hbs$.
    For the multiorder setting, \eqref{eq:interp_map_form} becomes
    \alna{
        \hbs
        &~=~&
        \argmin\nolimits_{\bbz}
        \frac{1}{\sigma_v^2}
        \| \bar{\bbs} - \bbTheta\bbz \|_2^2
        +
        \bbz^\top \bbD^2 \bbz
    &\nonumber\\&
        &~=~&
        \argmin\nolimits_{\bbz}
        \frac{1}{\sigma_v^2}
        \| \bar{\bbs} - \bbTheta\bbz \|_2^2
        +
        \sum_{k=0}^K
        \bbz^{k\top}
        \bbL_k
        \bbz^k,
    \label{eq:smooth_interp}}
    where $\bbz$ is a multiorder signal with $\bbz^k$ as the signal values over $k$-simplices.
    
    \im {\bf First-order AR:}
    We return to the scenario where $\bbs$ is a first-order AR process discussed in Section~\ref{S:AR_topo}, which has a precision matrix $\bbC^{-1} = \bbI + \alpha^2 \bbT^2 - 2\alpha \bbT$ for $\bbT \in \{ \bbL_k, \bbD \}$.
    In this setting, $\bbs$ follows a structural equation model (SEM) on the simplicial complex, where the signal at each simplex is modeled as a noisy linear combination of the signal at neighboring simplices.
    We rewrite~\eqref{eq:interp_map_form} as
    \alna{
        \hbs \!=\! \argmin\nolimits_{\bbz}
        \frac{1}{\sigma_v^2}
        \| \bar{\bbs} - \bbTheta\bbz \|_2^2
        +
        \bbz^{\!\top}\!
        \big(
            \bbI \!+\! \alpha^2 \bbT^2 \!-\! 2\alpha \bbT
        \big)
        \bbz.
    \label{eq:sem_interp}}
    When $\bbT = \bbD$, the term $\alpha^2 \bbD^2$ again promotes signal smoothness across all orders.
    The term $-2\alpha \bbD$ encourages consistent signal behavior between adjacent orders, where the orientation of the curl of $(k-1)$-signals does not greatly oppose the orientation of the gradient of $k$-signals.
    
    \im {\bf Subspace-aware:}
    We return to the setting in Section~\ref{S:subspace_psd}, where we know that $\bbs$ is primarily restricted to some subspace denoted by $\bbU_{\ccalS}$, a scenario common to certain domains.
    Recall that this corresponds to 
    \alna{
        \bbs \approx \bbU_{\ccalS} \tbs_{\ccalS},
        \quad
        \bbC \approx \bbU_{\ccalS} \diag(\bbp_{\ccalS}) \bbU_{\ccalS}^\top,
    \label{eq:subspace_model}}
    where $\ccalS$ not only indexes the active columns of the eigenvectors $\bbU$ of $\bbT$ but also selects the nonzero entries of the signal TFT $\tbs$ and the PSD $\bbp$.
    Then, $\bbC^{-1} = \bbU_{\ccalS} \diag^{-1}(\bbp_{\ccalS}) \bbU_{\ccalS}^\top$ in~\eqref{eq:interp_map_form}.
    However, observe that we can now rewrite~\eqref{eq:interp_obs_model} as
    \alna{
        \bar{\bbs} = \bbTheta \bbU_{\ccalS} \tbs_{\ccalS} + \bbv,
    \nonumber}
    which motivates solving for $\bbs$ in the frequency domain, similarly to Section~\ref{S:Wiener_topo}. 
    We then obtain the estimate
    \alna{
        \hat{\tbs}_{\ccalS}
        &~=~&
        \argmin\nolimits_{\tbz}
        \frac{1}{\sigma_v^2}
        \| \bar{\bbs} - \bbTheta \bbU_{\ccalS} \tbz \|_2^2
        +
        \tbz^\top
        \diag^{-1}(\bbp_{\ccalS}) 
        \tbz
    &\nonumber\\&
        &~=~&
        \Big(
            \bbU_{\ccalS}^\top \bbTheta^\top \bbTheta \bbU_{\ccalS}
            +
            \sigma_v^2 \diag^{-1} (\bbp_{\ccalS})
        \Big)^{-1}
        \bbU_{\ccalS}^\top \bbTheta^\top \bar{\bbs},
    \label{eq:bandlimit_interp}}
    from which we can obtain the signal estimate $\hbs = \bbU_{\ccalS} \hat{\tbs}_{\ccalS}$.

    In classical deterministic approaches, where we do not permit randomness and therefore $\sigma_v^2 \diag(\bbp_{\ccalS}) = \bbzero$ in~\eqref{eq:bandlimit_interp}, we recover the estimate $\hat{\tbs}$ while discarding elements of the partially observed signal $\bar{\bbs}$ that do not explicitly contribute to the active subspace $\bbU_{\ccalS}$.
    However, since we allow for stochasticity in~\eqref{eq:bandlimit_interp} due to both additive noise $\bbv$ and the variance of the stationary signal $\bbs$, we can account for the uncertainty in $\bar{\bbs}$ and obtain the most likely signal estimate $\hbs$ given the level of noise in the observations via $\sigma_v^2$ and the variance in the signal of interest via $\bbp_{\ccalS}$.
    Moreover, our stochastic approach enjoys further practicality since it is flexible to scenarios where the signal $\bbs$ is almost but not completely restricted to $\bbU_{\ccalS}$.
\end{itemize}
The above are practical choices of $\bbC$ for signal reconstruction from an arbitrary subset of observed simplices for either a single order or across all orders.
However, signal reconstruction in the multiorder setting $\bbT = \bbD$ is particularly appealing for the simplicial complex setting.
As mentioned above, we may only be able to measure the signal at certain orders.
For example, say we observe node and triangle signals on a $2$-order simplicial complex, but we cannot measure values on edge, that is, $\bbTheta$ selects only entries of $\bbs$ corresponding to $0$- and $2$-simplices.
Since we observe no edge signals, we cannot use edge connectivity via faces and cofaces to approximate the unknown values, so we are completely reliant on the signal at adjacent orders to decide what the edge signal values in $\bbs$ are.
This is only possible due to the Dirac operator $\bbD$ dictating the relationship between signal behavior across orders.

\subsection{Stationary Topological Signals under Gaussianity}
\label{S:GP_topo}

One important class of random topological signals is when the stationary signal $\bbs$ is Gaussian.
We specifically focus on the setting where the excitation $\bbw$ given to the topological filter $\bbH$ in~\eqref{eq:topo_stationarity_def} is Gaussian, but other models for Gaussian topological processes can be defined, such as~\cite{marinucci2025simplicial}.
Our case is analogous to the classical scenarios of Gaussianity for stationary processes in other domains, while we explore when the underlying domain is a simplicial complex.

If the input white noise $\bbw \sim \ccalN(\bbzero, \bbI)$ is Gaussian, then by the linearity in~\eqref{eq:topo_stationarity_def}, our signal $\bbs = \bbH \bbw$ is also Gaussian, where
\alna{
    \bbs \sim \ccalN(\bbzero, \bbC),
    \quad
    \bbC = \bbH^2,
\label{eq:pdf_gaussian_topo}}
and we recall that $\bbH = \sum_{r=0}^{R-1} h_r \bbT^r$ for the operator $\bbT \in \{ \bbL_k, \bbD \}$, which determines if we are interested in signals at single order $k$ or multiorder data.
Then, the PDF of $\bbs$ is
\alna{
    f_{\bbs}(\bbx)
    &~=~&
    \frac{ (2\pi)^{2/N} }{\det(\bbH) }
    \exp\bigg\{
        -\frac{1}{2} \bbx^\top \bbH^{-2} \bbx
    \bigg\},
\label{eq:pdf_topo_spatial}}
since $\bbC = \bbH^2$.
Moreover, by Definitions~\ref{D:WeaklyStationaryTopoProcess} and~\ref{def:topo_psd}, we recall that $\bbC = \bbU \diag(\bbp) \bbU^\top = \bbU\diag^2(h(\bblambda)\bbU^\top$ for filter frequency response $h(\lambda) = \sum_{r=0}^{R-1} h_r \lambda^r$, and $\bbU$ as the eigenvectors of $\bbT$.
Thus, as usual for stationary topological signals, the second-order statistical behavior of $\bbs$ is completely described by the PSD $\bbp$, and the simplicial complex structure is incorporated via $\bbU$.
However, more importantly, if $\bbs$ is Gaussian, then the frequency components of $\bbs$, that is, the entries of $\tbs = \bbU^\top \bbs$, are not only uncorrelated but also independent.
To see this, recall that Corollary~\ref{C:tft_cov} shows that $\tbs$ has the covariance $\tbC = \diag(\bbp)$.
Furthermore, the TFT is a linear operation, 
so $\tbs \sim \ccalN\big(\bbzero,\diag(\bbp)\big)$ also consists of Gaussian variables with PDF
\alna{
    f_{\tbs}(\tbx)
    =
    \prod_{i=1}^N
    \frac{1}{\sqrt{2\pi p_i}}
    \exp\bigg\{ -\frac{\tilde{x}_i^2}{ 2p_i } \bigg\}.
\label{eq:pdf_topo_freq}}
Thus, simplicial structure influences the covariance of $\bbs$ in the spatial domain via $\bbU$, but the frequency components $\tbs$ are independent with variances determined by $\bbp$.

When we want to estimate the covariance $\bbC$ from independent realizations $\bbS \in \reals^{N \times M}$, the sample covariance $\hbC$ in~\eqref{eq:sample_cov} is known to be the maximum likelihood estimate (MLE) of $\bbC$. 
However, we can also take the approach in Section~\ref{S:estimators} and obtain the most likely estimate of $\bbC$ while imposing prior structural information.
For example, we may consider Gaussian topological signals that follow an MA or AR model in~\eqref{eq:signal_MA} or~\eqref{eq:topo_AR_gen} for a parametric approach.
If $\bbC^{-1} \in \{ \bbL_k, \bbD^2 \}$, we may impose sparsity on $\bbC^{-1}$ for parsimonious simplex interdependencies or encourage a low-rank $\bbC^{-1}$ for underlying simplicity.
As one example, let $\bbs$ be an AR process, where $\big(\bbC^{\alpha}(\bbalpha)\big)^{-1} = \bbH^2(\bbalpha)$ as in~\eqref{eq:topo_AR_filter_form}.
Then, instead of obtaining the vanilla MLE as the sample covariance $\hbC$, we can consider the regularized optimization problem
\alna{
    \hbalpha
    ~=~
    & \argmin\nolimits_{\bbalpha} & ~~
    -\log \det( \bbTheta_{\bbalpha} )
    +
    \tr(\hbC \bbTheta_{\bbalpha} )
    +
    \lambda \ccalR(\bbTheta_{\bbalpha})
&\nonumber\\&
    &~~{\rm s.t.}~~& ~~
    \bbTheta_{\bbalpha} = \bbH^2(\bbalpha) \succ 0, ~~
    \bbTheta_{\bbalpha} \in \ccalD,
\label{eq:sc_lasso}}
where $\lambda > 0$ is some penalty weight, $\ccalR(\bbTheta_{\bbalpha})$ contains regularizers promoting desirable structural properties of $\bbTheta_{\bbalpha}$ such as sparsity, and $\ccalD$ ensures that $\bbTheta_{\bbalpha}$ follows known characteristics, such as having off-diagonal entries restricted to the supports of $\bbT$ and $\bbT^2$ if $R = 1$, in which case $\bbTheta_{\bbalpha} = \bbH^2(\bbalpha) = \bbI - 2 \alpha \bbT + \alpha^2 \bbT^2$.

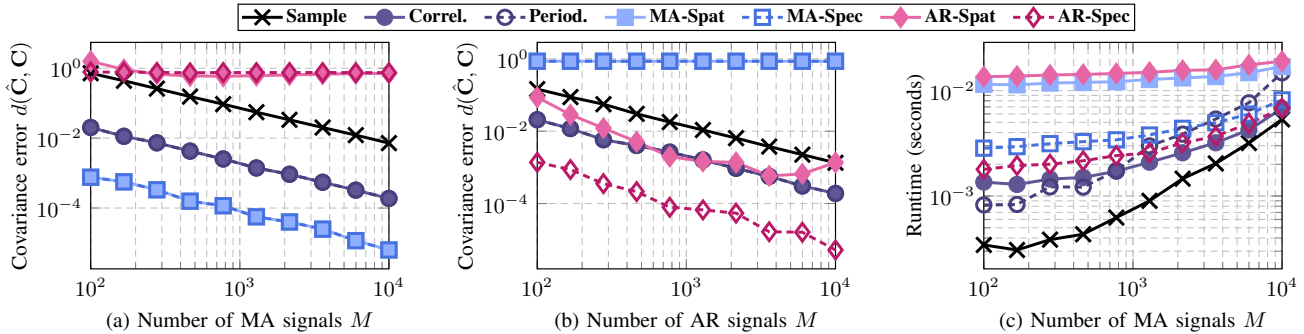
\begin{figure*}[t]
    \centering
    \begin{minipage}[b][][b]{0.30\textwidth}
        \centering
        \scalebox{.95}{\begin{tikzpicture}[baseline,scale=.95,trim axis left, trim axis right]

\pgfplotstableread{images/exp1_MA_table.csv}\errtable

\begin{loglogaxis}[
    xlabel={(a) Number of MA signals $M$},
    xmin=100,
    xmax=10000,
    ylabel={Covariance error $d(\hbC, \bbC)$},
    grid style=densely dashed,
    grid=both,
    legend style={
        at={(.5, 1.02)},
        anchor=south},
    legend columns=3,
    width=170,
    height=140,
    ]

    \addplot[clrsamp, mark=x, mark size=3pt, solid] 
        table [x=x, y=err_samp] {\errtable};
    
    \addplot[clrcorr, mark=*, mark size=2pt, solid] 
        table [x=x, y=err_cor] {\errtable};
    
    \addplot[clrperi, mark=o, mark size=2pt, densely dashed] 
        table [x=x, y=err_per] {\errtable};
    
    \addplot[clrmaspat, mark=square*, mark size=2pt, solid] 
        table [x=x, y=err_maspat] {\errtable};
    
    \addplot[clrmaspec, mark=square, mark size=2pt, densely dashed] 
        table [x=x, y=err_maspec] {\errtable};
    
    \addplot[clrarspat, mark=diamond*, mark size=2.5pt, solid] 
        table [x=x, y=err_arspat] {\errtable};
    
    \addplot[clrarspec, mark=diamond, mark size=2.5pt, densely dashed] 
        table [x=x, y=err_arspec] {\errtable};
    

\end{loglogaxis}

\end{tikzpicture}}
    \end{minipage} 
    \hspace{.1cm}
    \begin{minipage}[b][][b]{0.30\textwidth}
        \centering
        \scalebox{.95}{\begin{tikzpicture}[baseline,scale=.95,trim axis left, trim axis right]

\pgfplotstableread{images/exp1_AR_table.csv}\errtable

\begin{loglogaxis}[
    xlabel={(b) Number of AR signals $M$},
    xmin=100,
    xmax=10000,
    ylabel={Covariance error $d(\hbC, \bbC)$},
    grid style=densely dashed,
    grid=both,
    legend style={
        at={(.5, 1.05)},
        anchor=south},
    legend columns=7,
    width=170,
    height=140,
    ]

    \addplot[clrsamp, mark=x, mark size=3pt, solid] 
        table [x=x, y=err_samp] {\errtable};
    
    \addplot[clrcorr, mark=*, mark size=2pt, solid] 
        table [x=x, y=err_cor] {\errtable};
    
    \addplot[clrperi, mark=o, mark size=2pt, densely dashed] 
        table [x=x, y=err_per] {\errtable};
    
    \addplot[clrmaspat, mark=square*, mark size=2pt, solid] 
        table [x=x, y=err_maspat] {\errtable};
    
    \addplot[clrmaspec, mark=square, mark size=2pt, densely dashed] 
        table [x=x, y=err_maspec] {\errtable};
    
    \addplot[clrarspat, mark=diamond*, mark size=2.5pt, solid] 
        table [x=x, y=err_arspat] {\errtable};
    
    \addplot[clrarspec, mark=diamond, mark size=2.5pt, densely dashed] 
        table [x=x, y=err_arspec] {\errtable};
    
    \legend{ 
            {\bf Sample}~~~,
            {\bf Correl.}~~~,
            {\bf Period.}~~~,
            {\bf MA-Spat}~~~,
            {\bf MA-Spec}~~~,
            {\bf AR-Spat}~~~,
            {\bf AR-Spec},
            }

\end{loglogaxis}

\end{tikzpicture}}
    \end{minipage} 
    \hspace{.1cm}
    \begin{minipage}[b][][b]{0.30\textwidth}
        \centering
        \scalebox{.95}{\begin{tikzpicture}[baseline,scale=.95,trim axis left, trim axis right]

\pgfplotstableread{images/exp1_MA_table.csv}\errtable

\begin{loglogaxis}[
    xlabel={(c) Number of MA signals $M$},
    xmin=100,
    xmax=10000,
    ylabel={Runtime (seconds)},
    grid style=densely dashed,
    grid=both,
    legend style={
        at={(.5, 1.02)},
        anchor=south},
    legend columns=3,
    width=170,
    height=140,
    ]

    \addplot[clrsamp, mark=x, mark size=3pt, solid] 
        table [x=x, y=time_samp] {\errtable};
    
    \addplot[clrcorr, mark=*, mark size=2pt, solid] 
        table [x=x, y=time_cor] {\errtable};
    
    \addplot[clrperi, mark=o, mark size=2pt, densely dashed] 
        table [x=x, y=time_per] {\errtable};
    
    \addplot[clrmaspat, mark=square*, mark size=2pt, solid] 
        table [x=x, y=time_maspat] {\errtable};
    
    \addplot[clrmaspec, mark=square, mark size=2pt, densely dashed] 
        table [x=x, y=time_maspec] {\errtable};
    
    \addplot[clrarspat, mark=diamond*, mark size=2.5pt, solid] 
        table [x=x, y=time_arspat] {\errtable};
    
    \addplot[clrarspec, mark=diamond, mark size=2.5pt, densely dashed] 
        table [x=x, y=time_arspec] {\errtable};
    

\end{loglogaxis}

\end{tikzpicture}}
    \end{minipage} 
    \caption{
    Covariance estimation error as the number of observed signals $M$ increases for (a) MA topological signals and (b) AR topological signals.
    (c) Runtime comparison in seconds as the number of observed signals $M$ increases.
    }
    \label{fig:exp1_cov}
\end{figure*}

\subsection{Additional Directions}
\label{S:future_apps}

The examples discussed in this section are only a subset of the problems for which topological stationarity is either useful or interesting.
We share additional directions that warrant exploration in future work.

	\noindent\textbf{Time-varying topological signals.}  
    While many problems consider analyses and applications of static data, it is extremely common for data to vary over time.
    Thus, as one natural extension, we can implement our existing stationarity-based framework for time-dependent random topological signals.
    Potential adaptations for this setting include definitions of state-space representations and generalizations of dynamic filtering approaches for signals on simplicial complexes.

	\noindent\textbf{Robust covariance estimation.}  
    Another challenging yet realistic scenario occurs when we estimate the covariance of a stationary process from observations that are incomplete or corrupted by noise.
    In this work, we presented approaches to denoise or interpolate topological signals, which implies a simple two-step process, where we first estimate the covariance matrix and, then, use it to enhace our estimation.
    However, these steps are completely independent and may miss critical relationships that could be beneficial for estimation.
    This motivates \emph{robust} covariance estimators, particularly if we ensure that computation remains feasible or if we can handle common yet challenging scenarios, such as the presence of outliers.

	\noindent\textbf{Beyond simplicial complexes.}  
    While simplicial complexes certainly yield rich domains, we may also consider generalizing our framework to other higher-order structures, such as hypergraphs, cell complexes, or non-simplicial interactions.
    Recent progress in this direction extended notions of linear filtering to these settings~\cite{roddenberry2022signal,sardellitti2024topological}, which can naturally be followed by further generalizations of stationarity.
    However, without the structure imposed by the inclusion property, problems on more general objects often become ill-posed, so extensions would likely require much more careful analysis.
	
	\noindent\textbf{Simplicial complex estimation.}  
    A fundamental task in graph SP is estimating the graph structure from a set of observed signals, which includes methods that assume stationarity~\cite{dong2016learning, segarra2017network}.
    Thus, another natural direction is recovering a simplicial complex from stationary topological signals.
    For example, we could estimate the incidence matrices $\{ \bbB_k \}_{k=1}^K$, from which we can build the Dirac operator $\bbD$ or Hodge Laplacians $\{ \bbL_k \}_{k=0}^K$, or we can directly estimate $\bbD$ or $\bbL_k$.
    While the extension is straightforward conceptually, recovering the underlying simplicial structure would require much more consideration than the graph setting to ensure valid estimates, such as ensuring that $\bbB_k \bbB_{k+1} = \bbzero$.
    
	\noindent\textbf{Random simplicial complexes.}  
    Finally, even if we estimate or build a simplicial complex from data, we may not be certain that these reflect the underlying relationships among nodes.
    In addition to being able to construct simplicial complexes, we may often require the ability to model uncertainty not only in the signals but also the simplices connecting sets of nodes.
    Thus, another valuable direction for topological SP is probabilistic modeling for simplicial complexes.

\begin{figure*}[t]
    \centering
    \begin{minipage}[b][][b]{0.30\textwidth}
        \centering
        \scalebox{.95}{\begin{tikzpicture}[baseline,scale=.95,trim axis left, trim axis right]

\pgfplotstableread{images/exp2_cov_err_table.csv}\errtable

\begin{loglogaxis}[
    xlabel={(a) SNR (dB) \,\textendash\, MA signals},
    xmin=1,
    xmax=30,
    ylabel={Covariance error $d(\hbC, \bbC)$},
    grid style=densely dashed,
    grid=both,
    legend style={
        at={(.20, 1.05)},
        anchor=south west},
    legend columns=4,
    width=170,
    height=140,
    ]

    \addplot[clrsamp, mark=x, mark size=3pt, solid] 
        table [x=x, y=ma_samp] {\errtable};
    
    \addplot[clrcorr, mark=*, mark size=2pt, solid] 
        table [x=x, y=ma_corr] {\errtable};
    
    \addplot[clrperi, mark=o, mark size=2pt, densely dashed] 
        table [x=x, y=ma_peri] {\errtable};
    
    \addplot[clrmaspat, mark=square*, mark size=2pt, solid] 
        table [x=x, y=ma_maspat] {\errtable};
    
    \addplot[clrmaspec, mark=square, mark size=2pt, densely dashed] 
        table [x=x, y=ma_maspec] {\errtable};
    
    \addplot[clrarspat, mark=diamond*, mark size=2.5pt, solid] 
        table [x=x, y=ma_arspat] {\errtable};
    
    \addplot[clrarspec, mark=diamond, mark size=2.5pt, densely dashed] 
        table [x=x, y=ma_arspec] {\errtable};
    
    \addplot[clrwiener, mark=triangle*, mark size=2.5pt, solid] 
        table [x=x, y=ma_wiener] {\errtable};
    
    \legend{ 
            {\bf Sample}~~~,
            {\bf Correl.}~~~,
            {\bf Period.}~~~,
            {\bf MA-Spat}~~~,
            {\bf MA-Spec}~~~,
            {\bf AR-Spat}~~~,
            {\bf AR-Spec},
            {\bf Wiener},
            }

\end{loglogaxis}

\end{tikzpicture}}
    \end{minipage} 
    \hspace{.1cm}
    \begin{minipage}[b][][b]{0.30\textwidth}
        \centering
        \scalebox{.95}{\begin{tikzpicture}[baseline,scale=.95,trim axis left, trim axis right]

\pgfplotstableread{images/exp2_cov_err_table.csv}\errtable

\begin{loglogaxis}[
    xlabel={(b) SNR (dB) \,\textendash\, AR signals},
    xmin=1,
    xmax=30,
    ylabel={Covariance error $d(\hbC, \bbC)$},
    grid style=densely dashed,
    grid=both,
    legend style={
        at={(.55, 1.02)},
        anchor=south west},
    legend columns=4,
    width=170,
    height=140,
    ]

    \addplot[clrsamp, mark=x, mark size=3pt, solid] 
        table [x=x, y=ar_samp] {\errtable};
    
    \addplot[clrcorr, mark=*, mark size=2pt, solid] 
        table [x=x, y=ar_corr] {\errtable};
    
    \addplot[clrperi, mark=o, mark size=2pt, densely dashed] 
        table [x=x, y=ar_peri] {\errtable};
    
    \addplot[clrmaspat, mark=square*, mark size=2pt, solid] 
        table [x=x, y=ar_maspat] {\errtable};
    
    \addplot[clrmaspec, mark=square, mark size=2pt, densely dashed] 
        table [x=x, y=ar_maspec] {\errtable};
    
    \addplot[clrarspat, mark=diamond*, mark size=2.5pt, solid] 
        table [x=x, y=ar_arspat] {\errtable};
    
    \addplot[clrarspec, mark=diamond, mark size=2.5pt, densely dashed] 
        table [x=x, y=ar_arspec] {\errtable};
    
    \addplot[clrwiener, mark=triangle*, mark size=2.5pt, solid] 
        table [x=x, y=ar_wiener] {\errtable};
    

\end{loglogaxis}

\end{tikzpicture}}
    \end{minipage} 
    \hspace{.1cm}
    \begin{minipage}[b][][b]{0.30\textwidth}
        \centering
        \scalebox{.95}{\begin{tikzpicture}[baseline,scale=.95,trim axis left, trim axis right]

\pgfplotstableread{images/exp2_sig_err_table.csv}\errtable

\begin{loglogaxis}[
    xlabel={(c) SNR (dB)},
    xmin=1,
    xmax=30,
    ylabel={Signal error $d(\hbS, \bbS)$},
    grid style=densely dashed,
    grid=both,
    legend style={
        at={(.5, 1.05)},
        anchor=south},
    legend columns=2,
    width=170,
    height=140,
    ]

    \addplot[clrwiener, mark=*, mark size=2pt, solid] 
        table [x=x, y=ma_wiener] {\errtable};
    
    \addplot[clrnoisy, mark=square*, mark size=2pt, solid] 
        table [x=x, y=ma_noisy] {\errtable};
    
    \addplot[clrwiener, mark=o, mark size=2pt, densely dashed] 
        table [x=x, y=ar_wiener] {\errtable};
    
    \addplot[clrnoisy, mark=square, mark size=2pt, densely dashed] 
        table [x=x, y=ar_noisy] {\errtable};

    \legend{ 
            {\bf Wiener (MA)}~~~,
            {\bf Noisy (MA)}~~~,
            {\bf Wiener (AR)}~~~,
            {\bf Noisy (AR)},
            }

\end{loglogaxis}

\end{tikzpicture}}
    \end{minipage} 
    \caption{
    Covariance estimation error from noisy or filtered signals as SNR varies for (a) MA topological signals and (b) AR topological signals.
    (c) Denoising error as the SNR varies for both MA and AR topological signals.
    }
    \label{fig:exp2_cov}
\end{figure*}
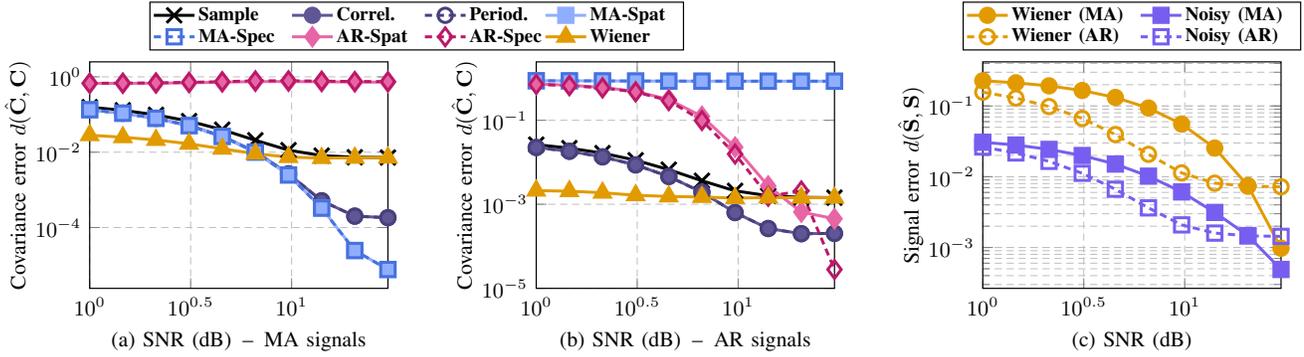

\section{Numerical results} \label{S:exps}

To complete the presentation of our topological stationarity framework and make concrete the introduced concepts, we provide numerical experiments showcasing covariance estimation and signal recovery.
For the following experiments, we measure error via $d(\hbX,\bbX) = \| \hbX - \bbX\|_F^2 / \| \bbX \|_F^2$, where $\hbX$ denotes an estimate and $\bbX$ is the target.
All results reflect median values of measurements over 50 independent trials.

\subsection{Covariance Estimation}
\label{sub:cov_est}

We first consider covariance estimation from stationary topological signals as discussed in Section~\ref{S:estimators}, where we observe $M$ signals $\bbS \in \reals^{N \times M}$, and we compare methods to estimate the corresponding covariance matrix $\bbC$.
For nonparametric methods, we consider 
the sample covariance in~\eqref{eq:sample_cov} (denoted {\bf Sample}), 
the correlogram estimator in~\eqref{eq:est_cov_cg} (denoted {\bf Cor})
and
the periodogram estimator in~\eqref{eq:est_cov_pg} (denoted {\bf Per}).
For parametric estimators, methods denoted {\bf MA} or {\bf AR} assume $\bbs$ is an MA or AR process, respectively.
For MA processes, {\bf MA-Spat} refers to estimation in the spatial domain~\eqref{eq:opt_gamma_topo} and {\bf MA-Spec} in the spectral domain~\eqref{eq:opt_gamma_freq_topo}.
Similarly, {\bf AR-Spat} and {\bf AR-Spec} represent the estimates~\eqref{eq:opt_eta_topo} and~\eqref{eq:opt_eta_freq_topo}.

We generate simplicial complexes of order $K=2$ with $N_0 = 50$ nodes, where we sample an edge for every pair of nodes with probability $0.2$ and a triangle for every triplet of edges with probability $0.3$.
From each synthetic simplicial complex, we generate multiorder stationary samples from either the MA model in~\eqref{eq:signal_MA} or the AR model~\eqref{eq:topo_AR_gen}, with $\bbT = \bbD$ and $\beta_r = \alpha_r = 0.1$ for every $r \in[R]$ and $R = 3$.
We estimate the covariance matrix $\bbC$ of the signals as the number of observations increases from $M = 10^2$ to $M = 10^4$, and we measure the resulting error (distance) per method, visualized for both signal models in Figs.~\ref{fig:exp1_cov}a \& \ref{fig:exp1_cov}b.
Unsurprisingly, methods that assume the correct signal model outperform other approaches; {\bf MA} and {\bf AR} methods achieve the lowest errors in Figs.~\ref{fig:exp1_cov}a  \& \ref{fig:exp1_cov}b, respectively.
Similarly, assuming an incorrect model attains higher error, with {\bf MA} approaches attaining the worst error for AR signals and vice versa.
However, while {\bf MA-Spat} and {\bf MA-Spec} yield equivalent solutions since the objectives of the problems~\eqref{eq:opt_gamma_topo} and~\eqref{eq:opt_gamma_freq_topo} return equivalent values, recall that~\eqref{eq:opt_eta_topo} and~\eqref{eq:opt_eta_freq_topo} do not.
Thus, we also find that {\bf AR-Spec} outperforms {\bf AR-Spat} since the spatial approach in~\eqref{eq:opt_eta_topo} requires far more samples for sufficient estimation relative to~\eqref{eq:opt_eta_freq_topo}.
Furthermore, recall that {\bf Cor} and {\bf Per} yield equivalent estimates, hence their identical performance in both Figs.~\ref{fig:exp1_cov}a  \& \ref{fig:exp1_cov}b, yet the periodogram estimate in~\eqref{eq:est_cov_pg} is more efficient for low-sample regimes.
To show this, we plot the runtime for each method in Fig.~\ref{fig:exp1_cov}c for MA signals, as the runtime comparison across methods is equivalent regardless of signal model.
As discussed in Section~\ref{Ss:nonparam_estimators}, {\bf Per} is faster than {\bf Cor} for low $M$, but when $M$ is large enough, the correlogram estimator {\bf Cor} becomes more efficient.
Moreover, as reflected in Fig.~\ref{fig:exp1_cov}c, recall that parametric estimation in the frequency domain as in {\bf MA-Spec} and {\bf AR-Spec} is more efficient than their spatial counterparts, despite the equivalent solutions of {\bf MA-Spat} and {\bf MA-Spec}.

\subsection{Topological Signal Denoising}\label{sub:TSD}

We consider signal denoising from Section~\ref{S:Wiener_topo}, where we wish to recover $M = 10^4$ multiorder signals $\bbS \in \reals^{N \times M}$ from noisy observations $\bbY = \bbS + \bbV$ with independent white Gaussian noise $V_{im} \sim \ccalN(0, \sigma_v^2)$ for different values of $\sigma_v^2$.
We generate simplicial complexes and their stationary data as before for MA and AR models, and we estimate the underlying covariance from the noisy signals $\bbY$ using all methods in Section~\ref{sub:cov_est}.
We further approximate the covariance via~\eqref{eq:sample_cov} from \emph{filtered} signals $\hbS$, obtained from Wiener filtering in~\eqref{eq:Wiener_spectral}.
Figs.~\ref{fig:exp2_cov}a  \& \ref{fig:exp2_cov}b plot the covariance estimation error for MA and AR signals, respectively, as the signal-to-noise ratio (SNR) increases, that is, $\sigma_v^2$ decreases. 
We observe that Wiener filtering is superior when the SNR is lower and $\bbY$ is noisier, as we can attenuate frequency components that are particularly noisy via~\eqref{eq:Wiener_spectral}, whereas even methods that assume the correct signal model can overfit to corrupted data.

We also directly measure the signal reconstruction error in Fig.~\ref{fig:exp2_cov}c, where we compare the error between $\bbY$ and $\bbS$ (denoted {\bf Noisy}) versus the error between $\hbS$ and $\bbS$ (denoted {\bf Wiener}), for both signal models.
As expected, Wiener filtering consistently returns a lower signal error and therefore improves our noisy observations.
Moreover, we also observe a larger gap between {\bf Wiener} and {\bf Noisy} for AR signals at higher SNR versus MA signals, as the AR model is more complex.

\subsection{Topological Signal Interpolation}\label{sub:TSI}

We next compare regularizers for topological signal interpolation from Section~\ref{S:reconstruction_topo}.
For a set of $M = 10^4$ multiorder simplicial signals $\bbS \in \reals^{N \times M}$, we observe noisy values for only a subset of $P < N$ simplices $\bar{\bbS} = \bbTheta \bbS + \bbV$, where $V_{im} \sim \ccalN(0,\sigma_v^2)$ for $\sigma_v^2 = 0.01$.
We sample MA signals with $R = 3$ and $\beta_r = 0.3$ for all $r \in [R]$, AR signals with $R = 1$ and $\alpha = 0.3$, and \emph{low-pass} stationary signals.
More specifically, we generate low-pass data via filtering as in~\eqref{eq:topo_stationarity_def} for $\bbH = \bbU\diag(h_{\rm lp}(\bblambda))\bbU^\top$ with $\bblambda$, $\bbU$ as the eigenvalues and eigenvectors of $\bbD$, respectively, and $h_{\rm lp} (\lambda) = (\lambda^2 + 10^{-3})^{-1}$.
We then perform topological signal recovery given the partial observations $\bar{\bbS}$ to compare methods in Section~\ref{S:reconstruction_topo}, where {\bf MAP} corresponds to~\eqref{eq:interp_closed_form} given the true covariance matrix $\bbC$, {\bf Smooth} solves~\eqref{eq:smooth_interp}, {\bf SEM} solves~\eqref{eq:sem_interp}, and {\bf Zero} assumes that $\bbC=\bbI$ and, as a result, implements $\bbTheta^\top \bar{\bbS}$, filling the unobserved signals with zeros.
In Fig.~\ref{fig:exp3_signal}, we show the signal reconstruction error for all three signal models across different ratios of observed simplices $P / N$.
We observe that for MA processes in Fig.~\ref{fig:exp3_signal}a, the methods {\bf Smooth}, {\bf SEM}, and {\bf Zero} with prior structural assumptions achieve a significantly higher error than {\bf MAP} with the true $\bbC$ since their assumptions do not capture the complexity of the order-$3$ topological filter in~\eqref{eq:topo_FIR_filter}.
However, in Fig.~\ref{fig:exp3_signal}b, the parametric approach {\bf SEM} given the parameter $\alpha \in \reals$ is equivalent to {\bf MAP} since the scalar $\alpha$ completely describes the AR signal covariance, but we need not invert $\bbC$ as in~\eqref{eq:interp_closed_form}.
Finally, for low-pass signals in Fig.~\ref{fig:exp3_signal}c, {\bf SEM} attains the highest error as it assumes a high-pass filter in~\eqref{eq:topo_AR_gen}, but {\bf Smooth} performs closest to the optimal {\bf MAP}, as {\bf Smooth} encourages more low-pass behavior in the estimated signals.

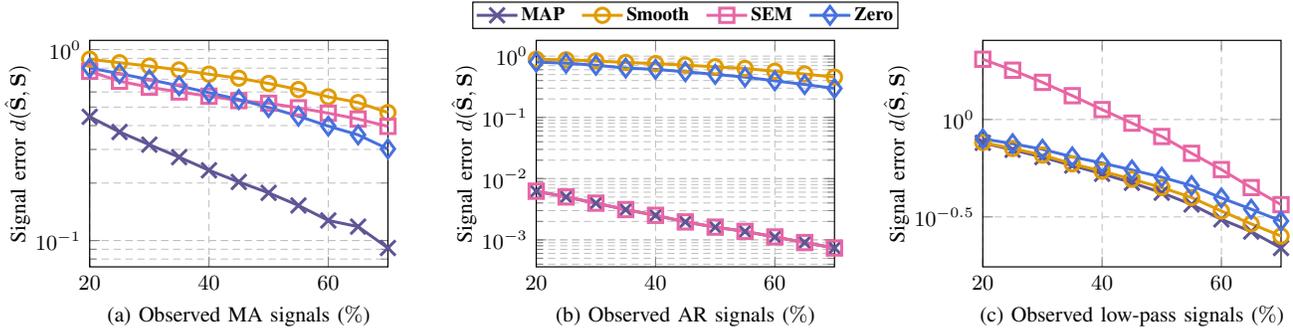
\begin{figure*}[t]
    \centering
    \begin{minipage}[b][][b]{0.30\textwidth}
        \centering
        \scalebox{.95}{\begin{tikzpicture}[baseline,scale=.95,trim axis left, trim axis right]

\pgfplotstableread{images/exp3_table.csv}\errtable

\begin{semilogyaxis}[
    xlabel={(a) Observed MA signals ($\%$)},
    xmin=20,
    xmax=70,
    ylabel={Signal error $d(\hbS, \bbS)$},
    grid style=densely dashed,
    grid=both,
    legend style={
        at={(.5, 1.02)},
        anchor=south},
    legend columns=4,
    width=170,
    height=140,
    ]

    \addplot[clrcorr, mark=x, mark size=3pt, solid] 
        table [x=x, y=ma_map] {\errtable};
    
    \addplot[clrwiener, mark=o, mark size=2pt, solid] 
        table [x=x, y=ma_smooth] {\errtable};
    
    \addplot[clrarspat, mark=square, mark size=2pt, solid] 
        table [x=x, y=ma_ar1] {\errtable};
    
    \addplot[clrmaspec, mark=diamond, mark size=2.5pt, solid] 
        table [x=x, y=ma_zero] {\errtable};
    

\end{semilogyaxis}

\end{tikzpicture}}
    \end{minipage} 
    \hspace{.1cm}
    \begin{minipage}[b][][b]{0.30\textwidth}
        \centering
        \scalebox{.95}{\begin{tikzpicture}[baseline,scale=.95,trim axis left, trim axis right]

\pgfplotstableread{images/exp3_table.csv}\errtable

\begin{semilogyaxis}[
    xlabel={(b) Observed AR signals ($\%$)},
    xmin=20,
    xmax=70,
    ylabel={Signal error $d(\hbS, \bbS)$},
    grid style=densely dashed,
    grid=both,
    legend style={
        at={(.5, 1.05)},
        anchor=south},
    legend columns=4,
    width=170,
    height=140,
    ]

    \addplot[clrcorr, mark=x, mark size=3pt, solid] 
        table [x=x, y=ar1_map] {\errtable};
    
    \addplot[clrwiener, mark=o, mark size=2pt, solid] 
        table [x=x, y=ar1_smooth] {\errtable};
    
    \addplot[clrarspat, mark=square, mark size=2pt, solid] 
        table [x=x, y=ar1_ar1] {\errtable};
    
    \addplot[clrmaspec, mark=diamond, mark size=2.5pt, solid] 
        table [x=x, y=ar1_zero] {\errtable};
    
    \legend{ 
            {\bf MAP}~~~,
            {\bf Smooth}~~~,
            {\bf SEM}~~~,
            {\bf Zero},
            }

\end{semilogyaxis}

\end{tikzpicture}}
    \end{minipage} 
    \hspace{.1cm}
    \begin{minipage}[b][][b]{0.30\textwidth}
        \centering
        \scalebox{.95}{\begin{tikzpicture}[baseline,scale=.95,trim axis left, trim axis right]

\pgfplotstableread{images/exp3_table.csv}\errtable

\begin{semilogyaxis}[
    xlabel={(c) Observed low-pass signals ($\%$)},
    xmin=20,
    xmax=70,
    ylabel={Signal error $d(\hbS, \bbS)$},
    grid style=densely dashed,
    grid=both,
    legend style={
        at={(.5, 1.02)},
        anchor=south},
    legend columns=4,
    width=170,
    height=140,
    ]

    \addplot[clrcorr, mark=x, mark size=3pt, solid] 
        table [x=x, y=smooth_map] {\errtable};
    
    \addplot[clrwiener, mark=o, mark size=2pt, solid] 
        table [x=x, y=smooth_smooth] {\errtable};
    
    \addplot[clrarspat, mark=square, mark size=2pt, solid] 
        table [x=x, y=smooth_ar1] {\errtable};
    
    \addplot[clrmaspec, mark=diamond, mark size=2.5pt, solid] 
        table [x=x, y=smooth_zero] {\errtable};
    

\end{semilogyaxis}

\end{tikzpicture}}
    \end{minipage} 
    \caption{
    Reconstruction error as the percentage of observed simplices increases for (a) MA topological signals, (b) AR topological signals, and (c) low-pass topological signals.
    }
    \label{fig:exp3_signal}
\end{figure*}

\begin{figure*}[b]
    \centering
    \begin{minipage}[b][][b]{0.30\textwidth}
        \centering
        \scalebox{.95}{\begin{tikzpicture}[baseline,scale=.95,trim axis left, trim axis right]

\pgfplotstableread{images/exp4adenoising_error_table.csv}\errtable

\begin{semilogyaxis}[
    xlabel={(a) SNR (dB) \,\textendash\, Keyword signals},
    xmin=1,
    xmax=10,
    ylabel={Signal error $d(\hbS, \bbS)$},
    grid style=densely dashed,
    grid=both,
    legend style={
        at={(.5, 1.05)},
        anchor=south},
    legend columns=3,
    width=170,
    height=140,
    ]

    \addplot[clrcorr, mark=x, mark size=3pt, solid] 
        table [x=x, y=errn] {\errtable};
    
    \addplot[clrwiener, mark=o, mark size=2pt, solid] 
        table [x=x, y=errsm] {\errtable};
    
    \addplot[clrarspat, mark=square, mark size=2pt, solid] 
        table [x=x, y=errwf] {\errtable};
    
    \legend{ 
            {\bf Noisy}~~~,
            {\bf Smooth}~~~,
            {\bf Wiener},
            }

\end{semilogyaxis}

\end{tikzpicture}}
    \end{minipage} 
    \hspace{.1cm}
    \begin{minipage}[b][][b]{0.30\textwidth}
        \centering
        \scalebox{.95}{\begin{tikzpicture}[baseline,scale=.95,trim axis left, trim axis right]

\pgfplotstableread{images/exp4ainterp_error_table.csv}\errtable


\begin{semilogyaxis}[
    xlabel={(b) Observed signals ($\%$)},
    xmin=20,
    xmax=70,
    ylabel={Signal error $d(\hbS, \bbS)$},
    grid style=densely dashed,
    grid=both,
    legend style={
        at={(.5, 1.05)},
        anchor=south},
    legend columns=3,
    width=170,
    height=140,
    ]

    \addplot[clrcorr, mark=x, mark size=3pt, solid] 
        table [x=x, y=err00] {\errtable};
    
    \addplot[clrwiener, mark=o, mark size=2pt, solid] 
        table [x=x, y=err01] {\errtable};
    
    
    \addplot[clrarspat, mark=square, mark size=2pt, solid] 
        table [x=x, y=err1] {\errtable};
    

    \legend{ 
            $\gamma\!=\!0$~~~,
            $\gamma\!=\!0.1$~~~,
            $\gamma\!=\!1$,
            }

\end{semilogyaxis}

\end{tikzpicture}}
    \end{minipage} 
    \hspace{.1cm}
    \begin{minipage}[b][][b]{0.30\textwidth}
        \centering
        \scalebox{.95}{\begin{tikzpicture}[baseline,scale=.95,trim axis left, trim axis right]

\pgfplotstableread{images/exp4aordinterp_error_table.csv}\errtable


\begin{semilogyaxis}[
    xlabel={(b) Observed $1$-signals ($\%$)},
    xmin=0,
    xmax=50,
    ylabel={Signal error $d(\hbS, \bbS)$},
    grid style=densely dashed,
    grid=both,
    legend style={
        at={(.5, 1.05)},
        anchor=south},
    legend columns=3,
    width=170,
    height=140,
    ]

    \addplot[clrcorr, mark=x, mark size=3pt, solid] 
        table [x=x, y=err00] {\errtable};
    
    \addplot[clrwiener, mark=o, mark size=2pt, solid] 
        table [x=x, y=err01] {\errtable};
    
    \addplot[clrarspat, mark=square, mark size=2pt, solid] 
        table [x=x, y=err05] {\errtable};
    

    \legend{ 
            $\gamma\!=\!0$~~~,
            $\gamma\!=\!0.1$~~~,
            $\gamma\!=\!0.2$,
            }

\end{semilogyaxis}

\end{tikzpicture}}
    \end{minipage} 
    \caption{
    (a) Denoising error as the SNR varies for keyword signals.
    (b) Reconstruction error as the percentage of observed simplices increases for keyword signals.
    (c) Reconstruction error as the percentage of observed $1$-simplices increases for keyword signals.
    }
    \label{fig:exp4_real}
\end{figure*}
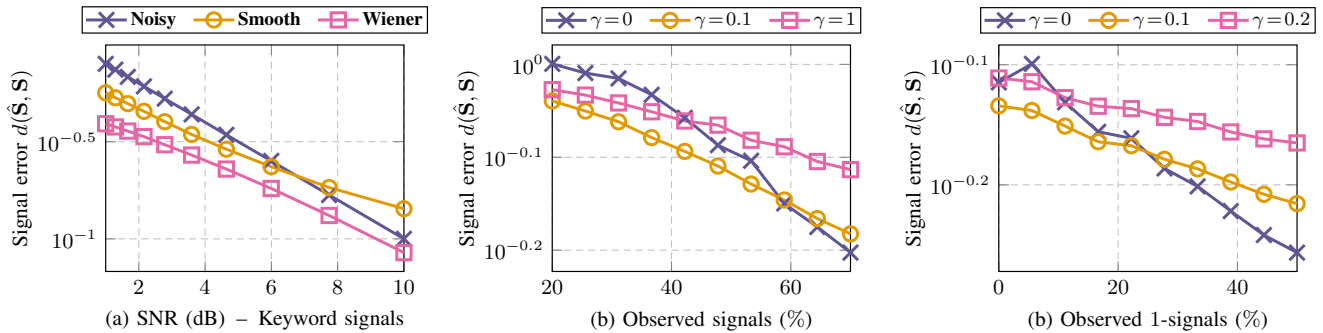

\subsection{Coauthorship Data}\label{sub:coautorship}

Finally, we consider a real-world setting with a simplicial complex of order $2$ that connects academic authors based on their collaborations.
Here, nodes represent authors, and edges and triangles connect researchers who have collaborated on at least one paper together.
We observe the average frequencies of using $M = 1903$ various keywords for each author as $0$-signals, and we build $1$- and $2$-signals by averaging the $0$-signals for authors connected by edges and triangles, respectively.
From this coauthorship data, we repeat simulations on signal denoising and interpolation for these keyword signals in Figure~\ref{fig:exp4_real} for $50$ independent trials each, where each trial corresponds to the random selection of an induced sub-simplicial complex of $N_0 = 50$ authors.

For Fig.~\ref{fig:exp4_real}a, we add white Gaussian noise $\bbV$ to the keyword signals $\bbS$ for corrupted signals $\bbY = \bbS + \bbV$ as in Section~\ref{sub:TSD} for different values of the noise variance $\sigma_v^2$.
We then compare the error between the noisy $\bbY$ and the true $\bbS$ versus the error of denoised signals $\hbS$ obtained via Wiener filtering~\eqref{eq:Wiener_spectral}.
Additionally, we compare to estimation under smoothness assumptions, that is, 
\alna{
    \hbS_{\rm smooth}
    =
    \argmin\nolimits_{\bbS}
    \| \bbS - \bbY \|_F^2 + \gamma \tr( \bbS^\top \bbD^2 \bbS ) 
\label{eq:smooth_denoising}}
for $\gamma = 0.1$.
We observe that smoothness can improve highly noisy observations, but without prior information dictating $\gamma$, the additional imposed structure can harm estimation.
Alternatively, the Wiener filtered signals $\hbS$ (denoted {\bf Wiener}) incorporate stochasticity in estimation and therefore achieve a lower error than both $\bbY$ (denoted {\bf Noisy}) and $\hbS_{\rm smooth}$ (denoted {\bf Smooth}) for all SNR values in Fig.~\ref{fig:exp4_real}a.

We also revisit signal interpolation by recovering the full signal values $\bbS$ given noisy observations on only a subset of simplices $\bar{\bbS} = \bbTheta\bbS + \bbV$ for $V_{im} \sim \ccalN(0,1)$, which we investigate as the ratio of observed simplices $P/N$ varies.
We interpolate $\bar{\bbS}$ by solving a modification of~\eqref{eq:interp_map_form}, 
\alna{
    \hbS
    =\!
    \argmin\nolimits_{\bbS}\!
    \| \bar{\bbS} - \bbTheta\bbS \|_F^2 + \tr(\bbS^\top \!\hbC^{-1}_{\rm pg} \bbS) + \gamma \tr(\bbS^\top \!\bbD^2 \bbS)
\label{eq:interp_stat_smooth}}
where we replace the unknown true covariance matrix $\bbC$ with $\hbC_{\rm pg}$ from~\eqref{eq:est_cov_pg}.
We also include a smoothness penalty and vary the weight $\gamma$ to demonstrate when the underlying signal smoothness helps estimation.
We again find that encouraging smoothness to an extent can improve estimation in Fig.~\ref{fig:exp4_real}b when a low ratio of simplices is observed, as $\gamma = 0.1$ attains a lower error than $\gamma \in \{ 0,1 \}$ when $P/N \leq 40$, but simply imposing stationarity $\gamma = 0$ is preferable when sufficiently many simplices are observed.

We also exemplify the scenario where we are missing data from $1$-simplices. 
We randomly drop varying numbers of $1$-signals for a set of $\bar{\bbS} \in \reals^{P \times M}$ partial observations with $P = N_0 + P_1 + N_2$ and $P_1 < N_1$.
Fig.~\ref{fig:exp4_real}c plots signal reconstruction error solving~\eqref{eq:interp_stat_smooth} for $\gamma \in \{0,0.1,0.2\}$, where $P_1 / N_1 = 0$ corresponds to the case where we recover $1$-signals without observing any edges.
We observe a similar trend as in Fig.~\ref{fig:exp4_real}b, where smoothness via $\gamma = 0.1$ can improve estimation when $P_1 / N_1$ is low, but too strong a penalty $\gamma = 0.2$ increases error.
Thus, prior structural knowledge can aid estimation when data is scarce, but the flexibility due to our stationarity assumption is sufficient when enough simplices are observed.

\section{Conclusion} \label{S:conclusion}

We introduced a probabilistic framework for random signals on simplicial complexes.
In particular, we defined stationarity topological processes to extend classical stationarity to a topological domain.
We explored the implications of stationarity for topological signals through spectral analyses of Hodge and Dirac operators, which led to the definition of the PSD of simplicial signals.
Thus, while our work follows the intuition of traditional SP models and techniques, we are able to model such data for random processes over complex, higher-order connections.
Furthermore, the algebraically friendly spectral properties of stationary topological processes allowed us to develop practical estimators of the signal covariance matrix and PSD, which are particularly useful in data-scarce regimes.
We also generalized classical tasks in SP such as signal denoising and interpolation for the simplicial complex domain.
We closed by discussing future directions, including extensions to more general topological structures and the application of our framework to generative or Bayesian models.

\bibliographystyle{IEEEtran}
\bibliography{mybib}

\begin{thebibliography}{10}
\providecommand{\url}[1]{#1}
\csname url@samestyle\endcsname
\providecommand{\newblock}{\relax}
\providecommand{\bibinfo}[2]{#2}
\providecommand{\BIBentrySTDinterwordspacing}{\spaceskip=0pt\relax}
\providecommand{\BIBentryALTinterwordstretchfactor}{4}
\providecommand{\BIBentryALTinterwordspacing}{\spaceskip=\fontdimen2\font plus
\BIBentryALTinterwordstretchfactor\fontdimen3\font minus \fontdimen4\font\relax}
\providecommand{\BIBforeignlanguage}[2]{{%
\expandafter\ifx\csname l@#1\endcsname\relax
\typeout{** WARNING: IEEEtran.bst: No hyphenation pattern has been}%
\typeout{** loaded for the language `#1'. Using the pattern for}%
\typeout{** the default language instead.}%
\else
\language=\csname l@#1\endcsname
\fi
#2}}
\providecommand{\BIBdecl}{\relax}
\BIBdecl

\bibitem{bick2023higher}
C.~Bick, E.~Gross, H.~A. Harrington, and M.~T. Schaub, ``What are higher-order networks?'' \emph{SIAM Rev.}, vol.~65, no.~3, pp. 686--731, 2023.

\bibitem{salnikov2018simplicial}
V.~Salnikov, D.~Cassese, and R.~Lambiotte, ``Simplicial complexes and complex systems,'' \emph{Eur. J. Physics}, vol.~40, no.~1, p. 014001, 2018.

\bibitem{barbarossa2020topological}
S.~Barbarossa and S.~Sardellitti, ``Topological signal processing over simplicial complexes,'' \emph{IEEE Trans. Signal Process.}, vol.~68, pp. 2992--3007, 2020.

\bibitem{schaub2021signal}
M.~T. Schaub, Y.~Zhu, J.-B. Seby, T.~M. Roddenberry, and S.~Segarra, ``Signal processing on higher-order networks: Livin’on the edge... and beyond,'' \emph{Signal Process.}, vol. 187, p. 108149, 2021.

\bibitem{isufi2024topological}
E.~Isufi, G.~Leus, B.~{Beferull-Lozano}, S.~Barbarossa, and P.~Di~Lorenzo, ``Topological signal processing and learning: Recent advances and future challenges,'' \emph{Signal Processing}, vol. 233, p. 109930, 2025.

\bibitem{hayes2009statistical}
M.~H. Hayes, \emph{{Statistical Digital Signal Processing and Modeling}}.\hskip 1em plus 0.5em minus 0.4em\relax John Wiley \& Sons, 2009.

\bibitem{stoica2005spectral}
P.~Stoica and R.~L. Moses, \emph{{Spectral Analysis of Signals}}.\hskip 1em plus 0.5em minus 0.4em\relax Prentice Hall Upper Saddle River, NJ, 2005.

\bibitem{girault2015stationary}
B.~Girault, ``Stationary graph signals using an isometric graph translation,'' in \emph{Eur. Signal Process. Conf. (EUSIPCO)}, 2015, pp. 1516--1520.

\bibitem{girault2015translation}
B.~Girault, P.~Gon{\c{c}}alves, and E.~Fleury, ``Translation on graphs: An isometric shift operator,'' \emph{IEEE Signal Process. Lett.}, vol.~22, no.~12, pp. 2416--2420, 2015.

\bibitem{EPFL16stationary}
N.~Perraudin and P.~Vandergheynst, ``Stationary signal processing on graphs,'' \emph{IEEE Trans. Signal Process.}, vol.~65, no.~13, pp. 3462--3477, 2017.

\bibitem{marques2017stationary}
A.~G. Marques, S.~Segarra, G.~Leus, and A.~Ribeiro, ``Stationary graph processes and spectral estimation,'' \emph{IEEE Trans. Signal Process.}, vol.~65, no.~22, pp. 5911--5926, 2017.

\bibitem{sandrymouraspg_tsp14freq}
A.~Sandryhaila and J.~M.~F. Moura, ``Discrete signal processing on graphs: Frequency analysis,'' \emph{IEEE Trans. Signal Process.}, vol.~62, no.~12, pp. 3042--3054, 2014.

\bibitem{segarra2017optimal}
S.~Segarra, A.~G. Marques, and A.~Ribeiro, ``Optimal graph-filter design and applications to distributed linear network operators,'' \emph{IEEE Trans. Signal Process.}, vol.~65, no.~15, pp. 4117--4131, 2017.

\bibitem{lim2020hodge}
L.-H. Lim, ``Hodge {Laplacians} on graphs,'' \emph{Siam Rev.}, vol.~62, no.~3, pp. 685--715, 2020.

\bibitem{yang2022simplicial}
M.~Yang, E.~Isufi, M.~T. Schaub, and G.~Leus, ``Simplicial convolutional filters,'' \emph{IEEE Trans. Signal Process.}, vol.~70, pp. 4633--4648, 2022.

\bibitem{isufi2022convolutional}
E.~Isufi and M.~Yang, ``Convolutional filtering in simplicial complexes,'' in \emph{IEEE Int. Conf. Acoust., Speech, Signal Process. (ICASSP)}, 2022, pp. 5578--5582.

\bibitem{schaub2018flow}
M.~T. Schaub and S.~Segarra, ``Flow smoothing and denoising: Graph signal processing in the edge-space,'' in \emph{IEEE Global Conf. Signal and Info. Process. (GlobalSIP)}, 2018, pp. 735--739.

\bibitem{jia2019graph}
J.~Jia, M.~T. Schaub, S.~Segarra, and A.~R. Benson, ``Graph-based semi-supervised \& active learning for edge flows,'' in \emph{ACM SIGKDD Int. Conf. Knowl. Discov. Data Min. (KDD)}, 2019, pp. 761--771.

\bibitem{liu2023unrolling}
C.~Liu, G.~Leus, and E.~Isufi, ``Unrolling of simplicial elasticnet for edge flow signal reconstruction,'' \emph{IEEE Open J. Signal Process.}, 2023.

\bibitem{bianconi2021topological}
G.~Bianconi, ``The topological {Dirac} equation of networks and simplicial complexes,'' \emph{J. Physics: Complexity}, vol.~2, no.~3, p. 035022, 2021.

\bibitem{baccini2022weighted}
F.~Baccini, F.~Geraci, and G.~Bianconi, ``Weighted simplicial complexes and their representation power of higher-order network data and topology,'' \emph{Physical Rev. E}, vol. 106, no.~3, p. 034319, 2022.

\bibitem{yang2022simplicialtrend}
M.~Yang and E.~Isufi, ``Simplicial trend filtering,'' in \emph{Asilomar Conf. Signals, Syst., Comput.}\hskip 1em plus 0.5em minus 0.4em\relax IEEE, 2022, pp. 930--934.

\bibitem{yang2022simplicialconvolutional}
M.~Yang, E.~Isufi, and G.~Leus, ``Simplicial convolutional neural networks,'' in \emph{IEEE Int. Conf. Acoust., Speech, Signal Process. (ICASSP)}, 2022, pp. 8847--8851.

\bibitem{battiloro2024generalized}
C.~Battiloro, L.~Testa, L.~Giusti, S.~Sardellitti, P.~Di~Lorenzo, and S.~Barbarossa, ``Generalized simplicial attention neural networks,'' \emph{IEEE Trans. Signal Inf. Process. Netw}, 2024.

\bibitem{krishnan2023simplicial}
J.~Krishnan, R.~Money, B.~Beferull-Lozano, and E.~Isufi, ``Simplicial vector autoregressive model for streaming edge flows,'' in \emph{IEEE Int. Conf. Acoust., Speech, Signal Process. (ICASSP)}, 2023, pp. 1--5.

\bibitem{reddy2024recovery}
S.~Reddy and S.~P. Chepuri, ``Recovery of signals on a simplicial complex from subsampled neighbourhood aggregation,'' \emph{IEEE Signal Process. Lett.}, 2024.

\bibitem{schaub2020random}
M.~T. Schaub, A.~R. Benson, P.~Horn, G.~Lippner, and A.~Jadbabaie, ``Random walks on simplicial complexes and the normalized {Hodge} 1-{Laplacian},'' \emph{SIAM Rev.}, vol.~62, no.~2, pp. 353--391, 2020.

\bibitem{marinucci2025simplicial}
L.~Marinucci, G.~D'Acunto, P.~Di~Lorenzo, and S.~Barbarossa, ``Simplicial {G}aussian models: Representation and inference,'' \emph{arXiv preprint arXiv:2510.12983}, 2025.

\bibitem{yang2025topological}
M.~Yang, ``Topological {S}chr\"odinger bridge matching,'' in \emph{Intl. Conf. on Learning Representations (ICLR)}, 2025.

\bibitem{yang2023hodge}
M.~Yang, V.~Borovitskiy, and E.~Isufi, ``Hodge-compositional edge {G}aussian processes,'' in \emph{Int. Conf. on Artif. Intell. and Stat.}, S.~Dasgupta, S.~Mandt, and Y.~Li, Eds., vol. 238.\hskip 1em plus 0.5em minus 0.4em\relax PMLR, 02--04 May 2024, pp. 3754--3762.

\bibitem{liu2025matched}
C.~Liu, V.~M. Tenorio, A.~G. Marques, and E.~Isufi, ``Matched topological subspace detector,'' \emph{IEEE Trans. Signal Process.}, 2026.

\bibitem{hatcher2002algebraic}
A.~Hatcher, \emph{{Algebraic Topology}}.\hskip 1em plus 0.5em minus 0.4em\relax Cambridge University Press, 2002.

\bibitem{calmon2022higher}
L.~Calmon, M.~T. Schaub, and G.~Bianconi, ``Higher-order signal processing with the {Dirac} operator,'' in \emph{Asilomar Conf. Signals, Syst., Comput.}, 2022, pp. 925--929.

\bibitem{roddenberry2022hodgelets}
T.~M. Roddenberry, F.~Frantzen, M.~T. Schaub, and S.~Segarra, ``Hodgelets: Localized spectral representations of flows on simplicial complexes,'' in \emph{IEEE Int. Conf. Acoust., Speech, Signal Process. (ICASSP)}.\hskip 1em plus 0.5em minus 0.4em\relax IEEE, 2022, pp. 5922--5926.

\bibitem{candes2015phase}
E.~J. Candes, X.~Li, and M.~Soltanolkotabi, ``Phase retrieval via {Wirtinger} flow: Theory and algorithms,'' \emph{IEEE Trans. Inf. Theory}, vol.~61, no.~4, pp. 1985--2007, 2015.

\bibitem{shechtman2015phase}
Y.~Shechtman, Y.~C. Eldar, O.~Cohen, H.~N. Chapman, J.~Miao, and M.~Segev, ``Phase retrieval with application to optical imaging: {A} contemporary overview,'' \emph{IEEE Signal Process. Mag.}, vol.~32, no.~3, pp. 87--109, 2015.

\bibitem{barabasi2013network}
A.-L. Barab{\'a}si, ``Network science,'' \emph{Philos. Trans. R. Soc. A}, vol. 371, no. 1987, p. 20120375, 2013.

\bibitem{kay1998fundamentals}
S.~Kay, \emph{{Fundamentals of Statistical Signal Processing: Detection Theory}}.\hskip 1em plus 0.5em minus 0.4em\relax Prentice-Hall, 1998.

\bibitem{roddenberry2022signal}
T.~M. Roddenberry, M.~T. Schaub, and M.~Hajij, ``Signal processing on cell complexes,'' in \emph{IEEE Int. Conf. Acoust., Speech, Signal Process. (ICASSP)}.\hskip 1em plus 0.5em minus 0.4em\relax IEEE, 2022, pp. 8852--8856.

\bibitem{sardellitti2024topological}
S.~Sardellitti and S.~Barbarossa, ``Topological signal processing over generalized cell complexes,'' \emph{IEEE Trans. Signal Process.}, 2024.

\bibitem{dong2016learning}
X.~Dong, D.~Thanou, P.~Frossard, and P.~Vandergheynst, ``Learning {Laplacian} matrix in smooth graph signal representations,'' \emph{IEEE Trans. Signal Process.}, vol.~64, no.~23, pp. 6160--6173, 2016.

\bibitem{segarra2017network}
S.~Segarra, A.~G. Marques, G.~Mateos, and A.~Ribeiro, ``Network topology inference from spectral templates,'' \emph{IEEE Trans. Signal Inf. Process. Netw}, vol.~3, no.~3, pp. 467--483, 2017.

\end{thebibliography}

\end{document}